%% file: OLC-LTV.tex
\documentclass[10pt, journal]{IEEEtran}  

\usepackage{color,amsmath,amsthm,amstext,amsfonts,amssymb, mathrsfs,mathtools,subfigure}
\usepackage{graphicx,algorithm}
\usepackage[algo2e,linesnumbered]{algorithm2e}
\usepackage{tikz}
\usepackage{setspace}
\usepackage{array} 
\usepackage{booktabs}  
\usepackage{bbm}
\usepackage{wasysym}
\usepackage{textcomp}
\usepackage{hyperref} 
\usepackage{caption}
\usepackage{graphicx} %
\usepackage[sort,compress]{cite}
\usepackage{cleveref}
\newtheorem{lemma}{Lemma}
\newtheorem{assumption}{Assumption}

\newtheorem{theorem}{Theorem} 
\newtheorem{definition}{Definition}
\newtheorem{corollary}{Corollary}

\theoremstyle{remark}
\newtheorem{remark}{Remark}

\SetKwInput{KwInput}{Input}

\newcounter{l1}
\newcounter{l2}
\newcounter{l3}
\newcommand{\bdotlist}{\begin{list}{$\bullet$}{}}
\newcommand{\bboxlist}{\begin{list}{$\Box$}{}}
\newcommand{\bbboxlist}{\begin{list}{\raisebox{.005in}{{\tiny $\blacksquare$ \ \ }}}{}}
\newcommand{\bdashlist}{\begin{list}{$-$}{} }
\newcommand{\blist}{\begin{list}{}{} }
\newcommand{\barablist}{\begin{list}{\arabic{l1}}{\usecounter{l1}}}
\newcommand{\balphlist}{\begin{list}{(\alph{l2})}{\usecounter{l2}}}
\newcommand{\bAlphlist}{\begin{list}{\Alph{l2}.}{\usecounter{l2}}}
\newcommand{\bdiamlist}{\begin{list}{$\diamond$}{}}
\newcommand{\bromalist}{\begin{list}{(\roman{l3})}{\usecounter{l3}}}

\providecommand{\norm}[1]{\lVert#1\rVert}

\newcommand{\beq}{\begin{equation}}
\newcommand{\eeq}{\end{equation}}

\let\[\equation
\let\]\endequation

\newcommand{\tn}{\textnormal}

\DeclarePairedDelimiterX{\Norm}[1]{\lVert}{\rVert}{#1}

\DeclareMathOperator*{\argmin}{arg\,min}

\allowdisplaybreaks

\title{Change Point Detection Approach for Online Control of Unknown Time Varying Dynamical Systems}

\author{Deepan Muthirayan, Ruijie Du, Yanning Shen, and Pramod P. Khargonekar
\thanks{This work is supported in part by the National Science Foundation under Grant ECCS-1839429 and ECCS-2207457.
Deepan Muthirayan, Ruijie Du, Yanning Shen and Pramod P. Khargonekar are with the Department of Electrical Engineering and Computer Sciences, University of California Irvine, Irvine, CA (emails: deepan.m@uci.edu, ruijied@uci.edu, yannings@uci.edu, pramod.khargonekar@uci.edu).}
}

\begin{document}

\maketitle
\thispagestyle{empty}
\pagestyle{empty}

\begin{abstract}
    We propose a novel change point detection approach for online learning control with full information feedback (state, disturbance, and cost feedback) for unknown time-varying dynamical systems.  We  show that our algorithm can achieve a sub-linear regret with respect to the class of Disturbance Action Control (DAC) policies, which are a widely studied class of policies for online control of dynamical systems, for any sub-linear number of changes and very general class of systems: (i) matched disturbance system with general convex cost functions, (ii) general system with linear cost functions. Specifically, a (dynamic) regret of $\Gamma_T^{1/5}T^{4/5}$ can be achieved for these class of systems, where $\Gamma_T$ is the number of changes of the underlying system and $T$ is the duration of the control episode. That is, the change point detection approach 
    achieves a sub-linear regret for any sub-linear number of changes, which other previous algorithms such as in \cite{minasyan2021online} cannot. 
    Numerically, we demonstrate that the change point detection approach is superior to a standard restart approach \cite{minasyan2021online} and  to standard online learning approaches for time-invariant dynamical systems. Our work presents the first regret guarantee for unknown time-varying dynamical systems in terms of a stronger notion of variability like the number of changes in the underlying system. The extension of our work to state and output feedback controllers is a subject of future work. 
\end{abstract}

\input{Introduction} 
\input{ProblemFormulation} 
\input{Algorithm-Results} 
\input{Numerical-Analysis}

\input{Conclusion}

\bibliographystyle{IEEEtran} 
\bibliography{Refs.bib}

\input{Appendix}

\end{document}

%% file: Introduction.tex
\section{Introduction}

In recent years, there has been significant interest in the finite-time performance of learning-based control algorithms for uncertain dynamical systems. Such a control setting is broadly termed as {\it online control}, borrowing the notion from online learning, where a learner's performance is assessed by their ability to learn from a finite number of samples. The performance in online control is typically measured  in terms of regret, which is the loss of  performance using the proposed algorithm as compared with the best possible policy. Predominantly, the goal is to design algorithms that adapt to uncertainties arising from disturbances and adversarial cost function so that the regret scales sub-linearly in $T$, i.e., as $T^\alpha$ with $\alpha < 1$, where $T$ is the duration of the control episode. Significant progress has been made in online control. For example, algorithms have been developed for control of unknown systems, with adversarial cost functions and disturbances \cite{dean2018regret, mania2019certainty, agarwal2019online, simchowitz2020improper}, algorithms for known systems with some predictability of future disturbances \cite{yu2020power, lin2021perturbation}, and for unknown systems with predictability \cite{muthirayan2021online}.

Control of uncertain systems is an extensively researched theme in control theory. Stochastic control,
robust control and adaptive control are large subfields with voluminous literature that address the analysis and synthesis of control for different types of uncertainties. In particular, adaptive control comes closest to ``online control'' described above. While the primary focus in adaptive control is on closed-loop stability and asymptotic performance,  there have been some papers on transient performance. Adaptive control has been studied for systems of all types such as linear, non-linear, and stochastic. There are many variants of adaptive control such as adaptive model predictive control, adaptive learning control, stochastic adaptive control, and robust adaptive control. These variations address the design of adaptive controllers for different variations of the basic adaptive control setting.  Thus, adaptive control is a very rich and extensively studied topic. The key differences in the ``online control'' setting from the classical adaptive control are (a) the consideration of regret as the measure of performance and (b) in some cases the more general nature of the costs, which could be adversarial and/or unknown. Consequently, the classical adaptive control approaches can be inadequate to analyze online control problems. From a techniques point of view, progress in online control is achieved by  merging tools from statistical learning, online optimization, and control theory.

A typical assumption in online control is that the system is time-invariant. In many circumstances, however, the underlying system or environment can be time-varying. While some works have studied time-varying dynamical systems \cite{han2022learning, baby2022optimal}, they have been limited to quadratic cost functions. Very recently, authors of \cite{minasyan2021online} explored the problem of online control of unknown time-varying linear dynamical systems for generic convex cost functions. Their work presents some impossibility results and a regret guarantee of $\widetilde{\mathcal{O}}\left(\vert I \vert \sigma_I + T^{2/3}\right)$ for any interval $I$, where $\vert I \vert$ denotes the length of the interval and $\sigma_I$ is the square root of the average squared deviation of the system parameters in the interval $I$. Clearly, in their case \cite{minasyan2021online}, the achievability of sub-linear regret is limited to scenarios with number of changes of the underlying system within $o(T^{1/3})$. Motivated by this observation, we investigate the question, {\it whether sub-linear regret is achievable for any number of changes over the duration $T$, and under what system, information and cost structures assumptions can we achieve sub-linear guarantees.} 

{\bf Contribution}: Distinct from most of prior works in online control, which study the control of time invariant dynamical systems, the present paper studies the problem of control of a time varying dynamical system over a finite time horizon for generic convex cost functions. Specifically, a linear dynamical system with arbitrary disturbances, whose system matrices can be time varying is considered. For such systems, we address the question of how to learn online and optimize when the system matrices are unknown, in addition to the cost functions and disturbances being arbitrary and unknown a priori. The goal is to design {\it algorithms with regret guarantees in terms of stronger notions of variability (compared to $\sigma$), such as the number of changes}. Towards this end, we consider the full information feedback structure, where in addition to the cost and state feedback at the end of a time step, the controller also receives disturbance as a feedback. We specifically consider the regret with respect to the class of Disturbance Action Control (DAC) policies \cite{minasyan2021online}, which are a widely used class of policies for online control of dynamical systems. 
 
 We propose a novel change point detection-based online control algorithm for unknown time-varying dynamical systems.  We  present guarantees for very general class of systems: (i) matched disturbance system with general convex cost functions, (ii) general system with linear cost functions. We show that, in both these settings, a (dynamic) regret of $\widetilde{\mathcal{O}}\left(\Gamma_T^{1/5} T^{4/5}\right)$ is achievable with a high probability, where $\Gamma_T$ is the number of times the system changes in $T$ time steps and $T$ is the duration of the control episode. Through numerical simulations, we demonstrate that the change point detection approach is superior to a standard restart approach, the adaptive algorithm of \cite{minasyan2021online}, and also standard online learning approach for time-invariant dynamical systems such as \cite{simchowitz2020improper}. Our result guarantees sub-linear regret for any sub-linear number of changes, which is an improvement over \cite{minasyan2021online} which cannot guarantee sub-linear regret for any number of changes. Our work presents the first regret guarantee in terms of a stronger notion of variability like the number of changes in the underlying system. The extension of our work to the setting without disturbance feedback is a subject of future work.  

{\it Notation}: We denote the spectral radius of a matrix $A$ by $\rho(A)$, the discrete time interval from $m_1$ to $m_2$  by $[m_{1}, m_{2}]$, and the sequence $(x_{m_1}, x_{m_1+1}, ..., x_{m_2})$ compactly by $x_{m_{1}:m_{2}}$. Unless otherwise specified, $\norm{\cdot}$ is the 2-norm of a vector and the Frobenious norm of a matrix. We use $\mathcal{O}(\cdot)$ for the standard order notation, and $\widetilde{\mathcal{O}}(\cdot)$ denotes the order neglecting the poly-log terms in $T$. We denote the inner product of two vectors $x$ and $y$ by $\langle x,y \rangle$.

%% file: ProblemFormulation.tex
\section{Problem Formulation}

We consider the online control of a general linear time-varying dynamical system. Let $t$ denote the time index, $x_t$, the state of the system, $y_t$, the output of the system that is to be controlled, $u_t$, the control input, $w_t$ and $e_t$, the disturbance and measurement noise, and $\theta_t = [A_t, B_{t}]$, the time-varying system matrices. Then, the equation governing the dynamical system is given by
\begin{align}
& x_{t+1} = A_t x_t + B_{t} u_t + B_{t,w} w_t, \nonumber \\
& y_t = C_t x_t + e_t.
\label{eq:syseq}
\end{align}
Let $w_t \in \mathbb{R}^q$, $e_t \in \mathbb{R}^p$, $x_t \in \mathbb{R}^n$, $y_t \in \mathbb{R}^p$, and $u_t \in \mathbb{R}^m$. We assume that the sequence of system parameters $\theta_{1:T}$ is unknown to the controller. The disturbance $w_t$ could arise from unmodeled dynamics and thus need not be stochastic. For generality, we assume that the disturbances and measurement noise are bounded and arbitrary. We denote the total duration of the control episode by $T$. 

Like in any control problem, at any time $t$, the controller incurs a cost $c_t(y_t,u_t)$, which is a function of the output and the control input. In addition to the system parameters being unknown, the sequence of cost functions $c_{1:T}$ and the disturbances $w_{1:T}$ for the duration $T$ is arbitrary and unknown a priori. We assume that the full cost function $c_t(\cdot, \cdot)$ and the disturbance $w_t$ are revealed to the controller after its action at $t$. Such a feedback is typical in online control and optimization and is termed the full information feedback. The difference here compared to a standard online control formulation is the feedback of the disturbance $w_t$. Thus, a control policy has the following information by any time $t$: (i) the cost functions and the disturbances till $t-1$, $c_{1:t-1}$ and $w_{1:t-1}$, (ii) the control inputs till $t-1$, $u_{1:t-1}$, and (iii) the observations till $t$, $y_{1:t}$. Let $\Pi_I$ denote the set of policies that satisfy this information setting. 

We denote a control policy by $\pi$. The state, output, and the control input under the policy is denoted by $x^\pi_t, y^\pi_t$ and $u^\pi_t$ respectively. Given that the cost functions and disturbances are only revealed incrementally, one step at a time, the control policy will have to be adapted online as and when the controller gathers information to achieve the best performance over a period of time. Like in a standard online control problem, we characterize the performance of a control policy over a finite time by its {\it regret}. We denote the regret of a policy $\pi$ over a duration $T$ with respect to a policy class $\Pi_\mathcal{M} \in \Pi_I$ 
 by $R_T(\pi)$:
\beq 
R_T(\pi) = \underbrace{\sum_{t=1}^T c_t(y^\pi_t, u^\pi_t)}_{\tn{Policy Cost}} - \underbrace{\min_{\kappa \in \Pi_\mathcal{M}} \sum_{t = 1}^T c_t(y^\kappa_t, u^\kappa_t)}_{\tn{comparator cost}}.  
\label{eq:reg-opt}
\eeq 
The primary goal is to design a control policy that minimizes the regret for the stated control problem. Since the regret minimization problem is typically hard, a typical goal is to design a policy that achieves sub-linear regret, i.e., a regret that scales as $T^\alpha$ with $T$, with a $\alpha < 1$ that is minimal. Such a regret scaling implies that the realized costs converge to that of the best policy from the comparator class asymptotically. Our objective is to design an adaptive policy that can track time variations and achieve sub-linear regret. We note that the regret defined above is static regret. Later, we present the extension to dynamic regret, which is a notion that is more suitable for time-varying dynamical systems.

 The comparator class we consider is the class of Disturbance Action Control (DAC) policies (see \cite{minasyan2021online}). A Disturbance Action Control (DAC) policy is defined as the linear feedback of the disturbances up to a certain history $h$. Let's denote a DAC policy by $\pi_{\mathrm{DAC}}$. Then, the control input $u^{\pi_{\mathrm{DAC}}}_t$ under policy $\pi_{\mathrm{DAC}}$ is given by
\beq 
u^{\pi_{\mathrm{DAC}}}_t = \sum_{k=1}^h M^{[k]}_tw_{t-k}.
\label{eq:drc-fk}
\eeq 
Here, $M_t = \left[M^{[1]}_t, \dots, M^{[h]}_t \right]$ are the feedback gains or the disturbance gains and are the (time-varying) parameters of $\pi_{\mathrm{DAC}}$. Here, we note that, $\pi_{\mathrm{DAC}}$ can be dynamic, i.e., it's parameters can be varying with time. Therefore, the regret defined in Eq. \eqref{eq:reg-opt} is the notion of dynamic regret. We note that the policy is implementable with disturbance feedback. Extension to the case without the disturbance feedback can be made by using estimates of the disturbances instead. We defer the treatment without any disturbance feedback to future work. Our objective here is to optimize the parameter $M$ online so that the regret with respect to the best DAC policy in hindsight is sub-linear.

 The DAC policy is typically used in online control for regulating systems with disturbances; see \cite{agarwal2019online}. The important feature of the DAC policy is that the optimization problem to find the optimal fixed disturbance gain for a given sequence of cost functions is a convex problem and is thus amenable to online optimization and online performance analysis. A very appealing feature of DAC is that, for time-invariant systems, the optimal disturbance action control for a given sequence of cost functions is very close in terms of the performance to the optimal linear feedback controller of the state; see \cite{agarwal2019online}. Thus, for time-invariant systems, by optimizing the DAC online, it is possible to achieve a sub-linear regret with respect to the best linear feedback controller of the state, whose computation is a non-convex optimization problem. 
 
For time-varying dynamical systems, as pointed out in \cite[Thoerem 2.1]{minasyan2021online}, there exist problem instances where the DAC class (with disturbance feedback) incurs a much better cost than other types of classes such as linear state or output feedback policies and vice versa. Therefore, the DAC class is not a weaker class to compete against compared to these standard classes. Moreover, as pointed out by the impossibility result \cite[Thoerem 3.1]{minasyan2021online}, it is an equally harder class to compete against in terms of regret. In this work, we focus our study on the regret minimization problem with respect to the DAC class (with disturbance feedback) and defer the treatment of other control structures to future work. 

Even with the disturbance feedback, the challenge of estimating the unknown system parameters does not diminish. This is because of the presence of measurement noise and the variations itself. In the time-invariant case, following an analysis similar to \cite{simchowitz2020improper}, it can be shown that, even with disturbance feedback, only a regret of $T^{2/3}$ can be achieved with the state-of-the-art methods, which is not any better than the regret that can be achieved without disturbance feedback (see \cite{simchowitz2020improper}). 
The same holds for the time-varying case. It can be shown that, what \cite{minasyan2021online} can achieve for the system in Eq. \eqref{eq:syseq}, even with disturbance feedback, cannot be improved. Therefore, the conclusions we draw later on comparing the bounds we derive and the regret upper bound of \cite{minasyan2021online} are valid. We state our other assumptions below. 

\begin{assumption}[System] 
(i) The system is stable, i.e., $\norm{C_{t+k+1} A_{t+k} \dots  A_{t+1}B_t}_2 \leq \kappa_a\kappa_b(1-\gamma)^{k}, ~~ \forall ~ k \geq 0, ~~ \forall ~ t$, where $\kappa_a > 0, \kappa_b > 0$ and $\gamma$ is such that $0 < \gamma < 1$, and where $\kappa_a, \kappa_b$ and $\gamma$ are constants. $B_t$ is bounded, i.e., $\norm{B_t} \leq \kappa_b$. (ii) The disturbance and noise $w_t$ and $e_t$ is bounded. Specifically, $\norm{w_t} \leq \kappa_w$, where $\kappa_w > 0$ is a
constant, and $\norm{e_t} \leq \kappa_e$, where $\kappa_e > 0$ is a constant.
\label{ass:sys}
\end{assumption}

\begin{assumption}[Cost Functions]
(i) The cost function $c_t$ is convex $\forall ~ t$.
(ii)  $\norm{c_t(x,u) -  c_t(x',u')} \leq LR\norm{z - z'}$ for a given $z^\top := [x^\top, u^\top], (z')^\top := [(x')^\top, (u')^\top]$, where $R:= \max\{\norm{z},\norm{z'},1\}$.
(iii) For any $d > 0$, when $\norm{x} \leq d$ and $\norm{u} \leq d$, $\nabla_x c(x,u) \leq Gd, \nabla_u c(x,u) \leq Gd$.
\label{ass:lipschitz}
\end{assumption}

\begin{remark}[System Assumptions]
Assumption \ref{ass:sys}.(i) is the equivalent of stability assumption used in time invariant systems. Such an assumption is typically used in online control when the system is unknown; see for eg., \cite{simchowitz2020improper, minasyan2021online}. Assumption \ref{ass:sys}.(iii) that noise is bounded is necessary, especially in the non-stochastic setting \cite{agarwal2019online, simchowitz2020improper}. The assumption on cost functions is also standard \cite{agarwal2019online}. 
\end{remark}

\begin{definition}
(i) $\mathcal{M}\! :=\!\! \{M = (M^{[1]},\dots,M^{[h]})\!:\! \norm{M^{[k]}} \leq \kappa_M \}$ (Disturbance Action Policy Class). (ii) $\mathcal{G} = \{ G^{[1:h]}: \norm{G^{[k]}}_2 \leq \kappa_a\kappa_b(1-\gamma)^{k-1} \}$. (iii) Setting (S-1): Matched disturbance system with convex cost functions: $B_{t} = B_{t,w}, C = I, e_t = 0$. Setting (S-2): General system with linear cost functions: $B_{t,w} = I$, and there exists a coefficient $\alpha_t \in \mathbb{R}^{p+m}$ such that $c_t(y,u) = \alpha^\top_tz$, $\norm{\alpha^\top_t} \leq G$.
\label{def:drc-class}
\end{definition}

%% file: Algorithm-Results.tex
\section{Online Learning Control Algorithm}
\label{sec:onlinelearning}

Typically, online learning control algorithms for time-invariant dynamical systems explore first for a period of time, and then exploit, i.e., adapt or optimize the control policy. While, in the time-invariant case, this strategy results in sub-linear regret, in the time-varying case, it can be less effective. For instance, consider the case where the system remains unchanged for the duration of the exploration phase and then changes around the instant when the exploration ends. Clearly, in this case, the estimate made at the end of the exploration phase will be very distant from the underlying system parameter realized after the exploration phase and therefore not result in a sub-linear regret.


We propose an online algorithm that {\it continuously learns to compute an estimate of the time varying system parameters and that simultaneously optimizes the control policy online}. Our estimation algorithm combines (i) a change point detection algorithm to detect the changes in the underlying system and (ii) a regular estimation algorithm. The online algorithm runs an online optimization parallel to the estimation to optimize the parameters of the control policy, which in our case is a DAC policy. 

{\bf Online Optimization}: Since the cost functions and the disturbances are unknown a priori, the optimal parameter $M$ of the DAC policy cannot be computed a priori. Rather, the parameters have to be adapted online continuously with the information gathered along the way to achieve the best performance. Given the convexity of the cost functions and the linearity of the system dynamics, we can apply the Online Convex Optimization (OCO) framework to optimize the policy parameters online. 
 
 We call a policy that learns the DAC policy parameters online as an online DAC policy. We formally denote such a policy by $\pi_{\mathrm{DAC-O}}$. Let the parameters estimated by $\pi_{\mathrm{DAC-O}}$ be denoted by $M_t = \left[M^{[1]}_t, \dots, M^{[h]}_t \right]$. Given that the parameter $M_t$ is continuously updated, the control input $u^{\pi_{\mathrm{DAC-O}}}_t$ can be computed by,
\beq 
u^{\pi_{\mathrm{DAC-O}}}_t = \sum_{k=1}^h M^{[k]}_t w_{t-k}.
\label{eq:drc-fk-o}
\eeq 
Given that the realized cost is dependent on the past control inputs, we will have to employ an extension of the OCO framework called Online Convex Optimization with Memory (OCO-M) to optimize the parameters of the DAC policy.

For the benefit of the readers, we briefly review the online convex optimization (OCO) setting (see \cite{hazan2008adaptive}). OCO is a game played between a player who is learning to minimize its overall cost and an adversary who is attempting to maximize the cost incurred by the player. At any time $t$, the player chooses a decision $M_t$ from some convex subset $\mathcal{M}$ given by $\max_{M \in \mathcal{M}} \norm{M} \leq \kappa_M$, and the adversary chooses a convex cost function $f_t(\cdot)$. As a result, the player incurs a cost $f_t(M_t)$ for its decision $M_t$. The goal of the player is to minimize the regret over a duration $T$, given by 
\begin{equation}
R_T  = \sum_{t=1}^{T} f_t(M_t) - \min_{M \in \mathcal{M}} \sum_{t=1}^{T} f_t(M). \nonumber 
\label{eq:regret-oco}
\end{equation}
The challenge is that the player does not know the cost function that the adversary will pick. Once the adversary picks a cost function, the player observes the realized cost and in some cases can also observe the full cost function. The objective of the learner is to achieve the minimal regret or at the least a sub-linear regret. We direct the readers to \cite{hazan2008adaptive} for a more detailed exposition and the various algorithmic approaches for this problem.

The difference in the OCO-M setting is that the cost functions can be dependent on the history of past decisions up to a certain time. Let the length of the history dependence be denoted by $h$. The regret in the OCO-M problem is then given by
\begin{equation}
R_T  = \sum_{t=1}^{T} f_t(M_{t-h:t}) - \min_{M \in \mathcal{M}} \sum_{t=1}^{T} f_t(M). \nonumber 
\label{eq:regret-oco-m}
\end{equation}

One limitation of the OCO-M framework is that it can only be applied when the length $h$ is fixed or bounded above. In a control setting though, the cost is typically a function of the state or the output, which is dependent on the full history of decisions $M_{1:t}$, the length of which grows unbounded with the duration of the control episode. Let 
\begin{align}
    & G_{t} = [G^{[1]}_{t}, G^{[2]}_{t}, \dots, G^{[h]}_{t}], \widetilde{G}_{t} = [\widetilde{G}^{[1]}_{t}, \widetilde{G}^{[2]}_{t}, \dots, \widetilde{G}^{[t-1]}_{t}], \nonumber \\
    & \widetilde{G}^{[k]}_t = C_t A_{t-1} \dots A_{t-k+2}A_{t-k+1}, ~ \forall ~  k \geq 2, ~ \widetilde{G}^{[1]}_t = C_t, \nonumber \\
    & G^{[k]}_t = C_t A_{t-1} \dots A_{t-k+2}A_{t-k+1}B_{t-k}, ~ \forall ~  k \geq 2, \nonumber 
\end{align}

and  $G^{[1]}_t = C_tB_{t-1}$. Thus, the history of dependence increases with $t$ and is not fixed. In order to apply the OCO-M framework, typically, a truncated  output $\tilde{y}_t$ is constructed, whose dependence on the history of control inputs is limited to $h$ time steps:
\begin{align}
& \tilde{y}^{\pi_{\mathrm{DAC-O}}}_t[M_{t:t-h} \vert G_t, s_{1:t}] = s_t + \sum_{k = 1}^{h} G^{[k]}_t u^{\pi_{\mathrm{DAC-O}}}_{t-k},  \nonumber \\
& ~ \tn{where} ~ s_t = y_t - \sum_{k=1}^{t-1} G^{[k]}_t u^{\pi_{\mathrm{DAC-O}}}_{t-k}. \nonumber 
\end{align} 

Using the truncated output, a truncated cost function $\tilde{c}_t$ is constructed as
\begin{align}
& \tilde{c}_t(M_{t:t-h}\vert G_t , s_{1:t}) \nonumber \\
& = c_t(\tilde{y}^{\pi_{\mathrm{DAC-O}}}_t[M_{t:t-h}\vert G_t , s_{1:t}],u^{\pi_{\mathrm{DAC-O}}}_t). \nonumber 
\end{align}

We denote the function $\tilde{c}_t(M_{t:t-h}\vert G_t , s_{1:t})$ succinctly by $\tilde{c}_t( M \vert G_t , s_{1:t})$ when each $M_k$ in $M_{t:t-h}$ is equal to $M$. This denotes the (truncated) cost that would have been incurred had the policy parameter been fixed to $M$ at all the past $h$ time steps. 

A standard gradient algorithm for OCO-M framework updates the decision $M_t$ by the gradient of the function $f_t(M_{t:t-h})$ with all $M_k$ in $M_{t:t-h}$ fixed to $M_t$. Using the same compact notation as above, this gradient is equal to $\partial f_t(M_{t})$. An interpretation of this gradient is that, it is the gradient of the cost that would have been incurred had the policy parameter been fixed at $M_t$ the past $h$ time steps. We employ the same idea to update the policy parameters of the DAC policy online. The online optimization algorithm we propose updates the policy parameter $M_t$ by the gradient of the cost function $\tilde{c}_t(M_t\vert G_t , s_{1:t})$ where each $M_k$ in $M_{t:t-h}$ is fixed to $M_t$, i.e., as 
\beq 
M_{t+1} = \tn{Proj}_{\mathcal{M}}\left(M_t - \eta \frac{ \partial \tilde{c}_t(M_t\vert G_t , s_{1:t})}{\partial M_t}\right), \label{eq:onlineudpate}
\eeq 
where $\mathcal{M}$ is a convex set of policy parameters.

\begin{definition}[Disturbance Action Policy Class]
$\mathcal{M} := \{M = (M^{[1]},\dots,M^{[h]}): \norm{M^{[k]}} \leq \kappa_M \}$
\label{def:drc-class}
\end{definition}

{\bf Optimization for Dynamic Regret}: The online optimization procedure described above can only fetch a sub-linear regret for static regret. To fetch a sub-linear dynamic regret, multiple online optimizers like in Eq. \eqref{eq:onlineudpate} are required to be run parallelly as in \cite{zhao2022non}. Let's index the parallel learners by $i$ and let the parameters corresponding to the learner $i$ be $M_{t,i}$. Just as in \cite{zhao2022non}, the final parameter $M_t$ is computed by $M_t = \sum_{i=1}^H p_{t,i} M_{t,i}$, where $p_{t,i}$ are a set of weights such that $\sum_{i=1}^N p_{t,i} = 1$ and $p_{t,i}$ are also updated online along with $M_{t,i}$s. Specifically, $p_{t,i}$ is updated by $p_{t+1,i} \propto p_{t,i} e^{-l_{t,i}(M_{t,i})}$, where $l_{t,i}(M) = \zeta \norm{M_{t,i} - M_{t-1,i}} + \langle M_{t,i},\partial \tilde{c}_t(M_t\vert G_t , s_{1:t})\rangle$. The $M_{t,i}$s are updated by
\beq 
M_{t+1,i} = \tn{Proj}_{\mathcal{M}}\left(M_{t,i} - \eta_i \frac{ \partial \tilde{c}_t(M_t\vert G_t , s_{1:t})}{\partial M_t}\right), 
\label{eq:onlineudpate-mul}
\eeq 
The complete online optimization algorithm is given in Algorithm \ref{alg:olc-fk}.

\LinesNumberedHidden{\begin{algorithm}[]
\DontPrintSemicolon
\KwInput{$\zeta, H$, Step sizes $\eta_i$s, parameters $\theta_{1:T}$.}
 
Initialize $M_{1,i} \in \mathcal{M}$ arbitrarily for all $i \in [1,H]$

Initialize $p_{1,i} \propto 1/(i^2+i)$ for all $i \in [1,H]$
 
  \For{t = 1,\dots,T}    
  { 
	
  Apply $u^{\pi_{\mathrm{DAC-O}}}_t = \sum_{k=1}^h M^{[k]}_t w_{t-k}$
	
  Observe $c_t, w_t$ and incur cost $c_t(y^{\pi_{\mathrm{DAC-O}}}_t,u^{\pi_{\mathrm{DAC-O}}}_t)$

  Compute: $l_{t,i} = \zeta \norm{M_{t,i} - M_{t-1,i}} + \langle M_{t,i},\partial \tilde{c}_t(M_t\vert G_t , s_{1:t})\rangle$ for all $i \in [1,H]$

  Update: $p_{t+1,i} \propto p_{t,i} e^{-l_{t,i}}$ for all $i \in [1,H]$
  
	
  Update: $M_{t+1,i} = \tn{Proj}_{\mathcal{M}}\left(M_{t,i} - \eta_i \frac{ \partial \tilde{c}_t(M_t\vert G_t , s_{1:t})}{\partial M_t}\right)$
  }

\caption{Online Learning Control with Full Knowledge (OLC-FK) Algorithm \cite[scream.control]{zhao2022non}  }
\label{alg:olc-fk}
\end{algorithm}
}
{\bf Main Result}: We state the performance of the algorithm OLC-FK formally below. 



\begin{theorem}[Full System Knowledge]
Suppose the setting is the general setting S-2, and the cost functions are general convex functions. Then, under Algorithm \ref{alg:olc-fk} \cite[scream.control]{zhao2022non}, with $h = \frac{\log{T}}{\left(\log{\left({1}/{1-\gamma}\right)}\right)}$, $H = \mathcal{O}(\log(T))$, $\zeta = \mathcal{O}(h^2)$ and $\eta_i = \mathcal{O}(2^{i-1}/\sqrt{\zeta T})$, the regret with respect to any DAC policy $M^\star_{1:T}$,
\[ R_T \leq \mathcal{O}\left(\sqrt{T(1+P_T)}\right), \] 
 where $P_T$ is the path length of the sequence $M^\star_{1:T}$.
\label{thm:olc-fk}
\end{theorem}

The proof follows from a standard proof for online optimization. Please see Appendix \ref{sec:th1-proof} for the full proof. 

\subsection{Disturbance Action Control without System Knowledge}

In the previous case, where the system parameters are known, the control policy parameters are optimized online through the truncated cost $\tilde{c}_t(\cdot)$, whose construction explicitly utilizes the knowledge of the underlying system parameters $G^{[k]}_t$. In this case, since the underlying system parameters are not available, we construct an estimate of the truncated state and the truncated cost by estimating the underlying system parameters $G^{[k]}_{t}$s. With this approach, the control policy will have to solve an online estimation problem to compute an estimate of the system parameters. Since the parameters are time-variant, the online estimation has to be run throughout, unlike the other online estimation approaches \cite{simchowitz2020improper, muthirayan2021online}, along with the policy optimization. Below, we describe in detail how our algorithm simultaneously performs estimation and optimizes the control policy. 

{\bf Online Estimation and Optimization}: The Online Learning Control with Zero Knowledge (OLC-ZK) of the system parameters has two components: (i) {\it a control policy} and (ii) {\it an online estimator} that runs in parallel to the control policy and throughout the control episode. The control policy and online optimization algorithm is similar to the online algorithm \ref{alg:olc-fk}, except that the control policy parameters are updated through an estimate of the truncated cost function. The online estimation algorithm employs a change point detection to identify the changes in the underlying system and a standard estimation algorithm to estimate the underlying system that is restarted after every detection of change. 
We discuss the details of our algorithm below.

{\bf A. Online Control Policy}: We use the same notation for the control policy and the control input, i.e., $\pi_{\mathrm{DAC-O}}$ and $u^{\pi_{\mathrm{DAC-O}}}_t$ respectively. The estimation algorithm constructs an estimate $\widehat{G}^{[k]}_t$ of the parameters $G^{[k]}_t$ of the system in Eq. \eqref{eq:syseq} for $k \in [1,h]$. Thus, the estimation algorithm estimates $G^{[k]}_t$s only for a truncated time horizon (looking backwards), i.e., for $k \in [1,h]$. We describe the estimation algorithm later. 

The policy $\pi_{\mathrm{DAC-O}}$ computes the control input $u^{\pi_{\mathrm{DAC-O}}}_t$ (zero knowledge case) by combining two terms: (i) {\it disturbance action control} just as in the full knowledge case and (ii) a {\it perturbation} for exploration. In this case, we require an additional perturbation, just as in \cite{dean2018regret}, so as to be able to run the estimation parallel to the Online DAC, the control for regulating the cost. Let $\tilde{u}^{\pi_{\mathrm{DAC-O}}}_t[M_{t} \vert w_{1:t}] = \sum_{k=1}^h M^{[k]}_t w_{t-k}$. Therefore, the total control input by $\pi_{\mathrm{DAC-O}}$ is given by
\beq 
u^{\pi_{\mathrm{DAC-O}}}_t = \underbrace{\tilde{u}^{\pi_{\mathrm{DAC-O}}}_t[M_{t} \vert w_{1:t}]}_{\tn{DAC}} + \underbrace{\delta u^{\pi_{\mathrm{DAC-O}}}_t}_{\tn{Perturbation}}.
\label{eq:cp-zk-o}
\eeq 

As in \cite{dean2018regret}, we apply a Gaussian random variable as the perturbation, i.e.,
\beq
\delta u^{\pi_{\mathrm{DAC-O}}}_t \sim \mathcal{N}(0,\sigma^2I),
\label{eq:pert-zk-o}
\eeq
where $\sigma$ denotes the standard deviation, and is a constant to be specified later.

In this case the policy parameters are optimized by applying OCO-M on an estimate of the truncated cost. To construct this estimate, we construct an estimate of $s_t$ and the truncated state $\tilde{x}^{\pi_{\mathrm{DAC-O}}}_t(\cdot)$. Given that $s_t$ is the state response when the control inputs are zero, we estimate $s_t$ by subtracting the contribution of the control inputs from the observed state:
\begin{align}
& \hat{s}_t = \sum_{k = 1}^{h} \widehat{G}^{[k]}_t w_{t-k} (\tn{S-1}) \nonumber \\
& \hat{s}_t = y^{\pi_{\mathrm{DAC-O}}}_t - \sum_{k = 1}^{h} \widehat{G}^{[k]}_t u^{\pi_{\mathrm{DAC-O}}}_{t-k} (\tn{S-2}).
\label{eq:natst-est}
\end{align}
Then the estimate of the truncated output follows by substituting $\hat{s}_t$ in place $s_t$ and using the estimated $\widehat{G}_t$ in place ${G}_t$:
\[ \tilde{\tilde{y}}^{\pi_{\mathrm{DAC-O}}}_t[M_{t:t-h} \vert \widehat{G}_t, \hat{s}_{1:t}] = \hat{s}_t + \sum_{k = 1}^{h} \widehat{G}^{[k]}_t \tilde{u}^{\pi_{\mathrm{DAC-O}}}_{t-k}. 
\]

Then, the estimate of the truncated cost is calculated as
\begin{align}
& \tilde{c}_t(M_{t:t-h}\vert \widehat{G}_t , \hat{s}_{1:t})\nonumber \\
& = c_t(\tilde{\tilde{y}}^{\pi_{\mathrm{DAC-O}}}_t[M_{t:t-h}\vert \widehat{G}_t , \hat{s}_{1:t}],\tilde{u}^{\pi_{\mathrm{DAC-O}}}_t). \nonumber 
\end{align}

The online update to the policy parameters is just as in Algorithm \ref{alg:olc-fk}, i.e., by the gradient of the estimate of the truncated cost
\[ 
M_{t+1,i} = \tn{Proj}_{\mathcal{M}}\left(M_{t,i} - \eta_i \frac{ \partial \tilde{c}_t(M_t\vert \widehat{G}_t , \hat{s}_{1:t})}{\partial M_t}\right).  
\]

{\bf A. Online Estimation}: The online estimation algorithm is a combination of a change point detection algorithm and a regular estimation algorithm. The change point detection algorithm detects changes larger than a certain threshold and resets the estimation algorithm upon every detection. The estimation algorithm continuously updates the estimates using all of the data from the last reset point. This offers the online learner more flexibility as it only resets whenever there is a significant underlying change, while it continues to refine the estimate otherwise. Thus, the change point detection approach can track the time variations more optimally. This is observed to be the case in the numerical simulations. 

{\bf A.1. Change Point Detection}: The goal of the Change Point Detection (CPD) algorithm is to detect the underlying changes in the system reliably. To do this, the CPD algorithm runs a sequence of independent estimation algorithms each of duration $t_p = N+h$ one after the other, where the estimation algorithms are the standard least-squares estimation applied to the data collected from the respective periods of duration $t_p = N+h$. Here, $t_p$ has to be necessarily greater than $h$, since computing the estimate of $G_t$ requires at least a length of $h$ inputs. Essentially, the CPD algorithm ignores the past and only considers the recent history to compute an estimate of the system parameters. This allows the CPD algorithm to compute a reliable estimate of the current values of the parameters of the system provided $N$ is of adequate size and at the same time not very large. Then, provided the estimation in each period of duration $N$ is an accurate estimate of the system parameter values in the respective periods, any change point can be detected by comparing the estimates across the different periods. More specifically, if the estimate at the end of a period is greater than a certain threshold compared to the estimate from an earlier period, we can proclaim change point detection. 

We denote the index of the successive periods of duration $t_p$ by $k$. We denote the start and end time of each of these periods by $t^k_s$ and $t^k_e$. Therefore, it follows that $t^k_s = t^{k-1}_e$ for all $k$. The CPD algorithm computes the following least-squares estimate at the end of each period $k$
\begin{align}
& \widehat{G}^{\mathrm{cd}}_k = \argmin_{\widehat{G}} \sum_{p = t^k_s+h}^{t^k_e} \ell_p\left(\widehat{G}\right) +  \lambda \norm{\widehat{G}}^2, \nonumber \\
& \lambda > 0 , ~ \ell_p\left(\widehat{G}\right) = \norm{y^{\pi_{\mathrm{DAC-O}}}_p - \sum_{l=1}^h \widehat{G}^{[l]}\delta u^{\pi_{\mathrm{DAC-O}}}_{p-l}}^2. 
\label{eq:cd-est-k}
\end{align}

We denote the first period of duration $t_p$, after a detection, as the baseline period with index $k=1$. By default, the very first period of duration $t_p$ at the beginning of the control episode is also a period with index $k=1$. The CPD algorithm {\bf proclaims change point detection, when at the end of a period $k$}
\beq
\norm{\widehat{G}^{\mathrm{cd}}_k - \widehat{G}^{\mathrm{cd}}_\ell}_2 > \frac{2\beta}{\sigma\sqrt{N}}, ~~ \tn{for any} ~~ \ell ~~ \tn{s.t.} ~~ 1 \leq \ell < k. \nonumber 
\eeq 
where $\beta$ is a constant to be defined later.

\LinesNumberedHidden{\begin{algorithm}[h]
\DontPrintSemicolon
\KwInput{Step sizes $\eta_{1:H}$, $H, \zeta, \sigma, \beta, N, h$}
 
Initialize $M_{1,i} \in \mathcal{M}$ arbitrarily $\forall ~ i \in [1,H]$, $t_d = 1, k = 1, t_s = 1, t_e = N+h$.

Initialize $p_{1,i} \propto 1/(i^2+i)$ for all $i \in [1,H]$.
 
  \For{t = 1,\dots,T}{ 
   Observe $y^{\pi_{\mathrm{DAC-O}}}_t$.
   
  \If{$t == t_e$}{
  Estimate $\widehat{G}^{\mathrm{cd}}_k$ according to Eq. \eqref{eq:cd-est-k}.
   
  \eIf{$k > 1$}{
  \eIf{$\norm{\widehat{G}^{\mathrm{cd}}_k - \widehat{G}^{\mathrm{cd}}_\ell}_2 > \frac{(2)\beta}{\sigma\sqrt{N}}$ ~ for any $1 \leq \ell < k$}{
  Proclaim change point detection. Set $t_d = t$. Set $k = 1$.}{$k=k+1$.}}
  {$k=k+1$.}
   
  $t_s = t_e, ~~ t_e = t_s+N+h-1$.
   
  }
   Compute $\widehat{G}_t$ according to Eq. \eqref{eq:ls-est}.
   
   Apply $u^{\pi_{\mathrm{DAC-O}}}_t$ from Eq. \eqref{eq:cp-zk-o}.
	
   Observe $c_t, w_t$ and incur cost $c_t(y^{\pi_{\mathrm{DAC-O}}}_t,u^{\pi_{\mathrm{DAC-O}}}_t)$.

   Compute: $l_{t,i} = \zeta \norm{M_{t,i} - M_{t-1,i}} + \langle M_{t,i},\partial \tilde{c}_t(M_t\vert \widehat{G}_t , \hat{s}_{1:t})\rangle$ for all $i \in [1,H]$

   Update: $p_{t+1,i} \propto p_{t,i} e^{-l_{t,i}}$ for all $i \in [1,H]$
	
   Update: $M_{t+1,i} = \tn{Proj}_{\mathcal{M}}\left(M_{t,i} - \eta_i \frac{ \partial \tilde{c}_t(M_t\vert \widehat{G}_t , \hat{s}_{1:t})}{\partial M_t}\right)$.
   }

\caption{Online Learning Control with Change Point Detection (OLC-ZK-CPD) Algorithm }
\label{alg:olc-zk-cd}
\end{algorithm}}

{\bf A.2. System Estimation}: Upon detection of a change by the CPD algorithm, the online estimation algorithm restarts the estimation of the system parameters after a delay of $h$. Let $t_d$ denote the most recent time of detection by the CPD algorithm. Then, the estimate of the system parameters for any time $t \geq t_d + 2h$ is given by
\begin{align}
& \widehat{G}_t = \tn{Proj}_{\mathcal{G}}(\widehat{G}^{\star}_t), ~~ \widehat{G}^{\star}_t = \argmin_{\widehat{G}} \sum_{p = t_d+h}^{t - h} \ell_p\left(\widehat{G}\right) +  \lambda \norm{\widehat{G}}^2, \nonumber \\
& \ell_p\left(\widehat{G}\right) = \norm{y^{\pi_{\mathrm{DAC-O}}}_p - \sum_{l=1}^h \widehat{G}^{[l]}\delta u^{\pi_{\mathrm{DAC-O}}}_{p-l}}^2.
\label{eq:ls-est}
\end{align}

{\bf Main Result}: The complete algorithm for the unknown system case is shown in Algorithm \ref{alg:olc-zk-cd}. We state the performance of the algorithm OLC-ZK-CPD formally below. 

\begin{definition}[Parameters]
$h = \frac{\log{T}}{\log\left({1}/(1-\gamma\right)},  
$ and $ \beta =\! 2\sqrt{h}\zeta_\Delta\left( \!\!\sqrt{ n\log\left(2\right) + 2\log\left(\frac{2h}{\delta}\right)} \!\!+ \!\frac{\lambda\kappa_a\kappa_b}{\gamma\zeta_\Delta\sigma\sqrt{hN}}\!\!\right)$,
where $\zeta_\Delta = \left(R_s + \frac{\kappa_a\kappa_b\kappa_m\kappa_wh}{\gamma} + \frac{\kappa_a\kappa_bR_u}{\gamma}\right), R_u = \kappa_M\kappa_wh + 3\sigma\sqrt{m+\log(1/\delta)}$, $R_s = \frac{\kappa_a\kappa_w}{\gamma} + \kappa_e + \frac{2R_u\kappa_a\kappa_b}{\gamma}$ and $\delta > 0$ is a constant. $N = \Gamma^{-0.8}_TT^{4/5}, ~ \sigma = \Gamma^{0.2}_T T^{-1/5}$.
\label{def:par-setting}
\end{definition}

\begin{theorem}[Zero System Knowledge]
 Consider Algorithm \ref{alg:olc-zk-cd} with the parameters given by Definition \ref{def:par-setting}. Suppose $T \geq 3$, $H = \mathcal{O}(\log(T))$, $\zeta = \mathcal{O}(h^2)$, $\eta_i = \mathcal{O}(2^{i-1}/\sqrt{\zeta T})$,  
 $\lambda \propto \mathcal{O}(1)$, $\Gamma_T = \widetilde{O}(T^d), ~ d < 1$ and the setting is either S-1 or S-2. Then, for $\delta \leq 1/T$, $\widetilde{\delta}$ arbitrarily small and $\delta \leq \widetilde{\delta}$, the regret with respect to any DAC policy $M^\star_{1:T}$, 
$ R_T \leq \widetilde{\mathcal{O}}\left(\sqrt{T(1+P_T)} + \Gamma^{1/5}_TT^{4/5}\right)$
 with probability greater than $1-\widetilde{\delta}$, where $P_T$ is the path length of the sequence $M^\star_{1:T}$. 
\label{thm:olc-zk-cd}
\end{theorem}

Please see Appendix \ref{sec:th2-proof} and \ref{sec:th2-proof-1} for the full proof.

\begin{definition}[Switching DAC Policy]
We defining a switching DAC policy as a policy which shifts its policy parameter $M$ at the instances of change in the underlying system.
\end{definition}

\begin{corollary}[Best Switching Policy]
Suppose the setting is either S-1 or S-2. Then, under the parameter setting of Theorem \ref{thm:olc-zk-cd}, for any $\delta \leq 1/T$, $\widetilde{\delta}$ arbitrarily small and $\delta \leq \widetilde{\delta}$, the regret with respect to the best switching policy $M^\star_{1:T}$, 
$ R_T \leq \widetilde{\mathcal{O}}\left(\Gamma^{1/5}_TT^{4/5}\right)$
 with probability greater than $1-\widetilde{\delta}$. 
\label{cor:olc-zk-cd}
\end{corollary}

This is a straightforward conclusion that follows from Theorem \ref{thm:olc-zk-cd} after recognizing the fact that the number of switches of the switching policy is $\Gamma_T$. 


\begin{remark}[Regret Result]
Minasyan et al. \cite{minasyan2021online} prove an adaptive regret bound of $\widetilde{\mathcal{O}}\left(\vert I \vert \sigma_I + T^{2/3}\right)$ for any interval $I$ of length $\vert I \vert$, where $\sigma_I$ is the square root of the average squared deviation of $G_t$ over the interval $I$. The key difference compared to \cite{minasyan2021online} is that our result is sub-linear with respect to the number of changes $\Gamma_T$ instead of $\sigma$, and we present a dynamic regret bound that is $\widetilde{\mathcal{O}}\left(\sqrt{T(1+P_T)} + \Gamma^{1/5}_TT^{4/5}\right)$. To compare with \cite{minasyan2021online}, lets consider the best switching policy corresponding to the switches in the underlying system. Let $I$ be any interval where the system does not change and let $M^\star_k$ correspond to the best policy parameter for the interval $k$. Then, the regret achieved by \cite{minasyan2021online} with respect to $M^\star_{1:\Gamma_T}$ is $\Gamma_T T^{2/3}$. The regret achieved by our algorithm is $\widetilde{\mathcal{O}}\left(\Gamma^{1/5}_TT^{4/5}\right)$, which follows from the fact that $P_T = \mathcal{O}(\Gamma_T)$. It follows that, we can achieve a sub-linear regret guarantee for $\Gamma_T = \widetilde{O}(T^d)$ for any $d < 1$, whereas the achievability of sub-linear regret in \cite{minasyan2021online} is limited to scenarios with $\Gamma_T = o(T^{1/3})$. 
\end{remark}

\begin{remark}[Unknown Time Variation]
Algorithm \ref{alg:olc-zk-cd} assumes the knowledge of total number of changes. We can extend our algorithm to the unknown time variation case by learning the optimal interval period $N$ and optimal $\sigma$ from an ensemble by using a meta-bandit algorithm on top of Algorithm \ref{alg:olc-zk-cd} just as in \cite{zhao2020simple}. 
We plan to incorporate this in our journal version.
\end{remark}


%% file: Numerical-Analysis.tex
\begin{figure}
\centering
\subfigure[]{%
\centering
  \includegraphics[scale=.2]{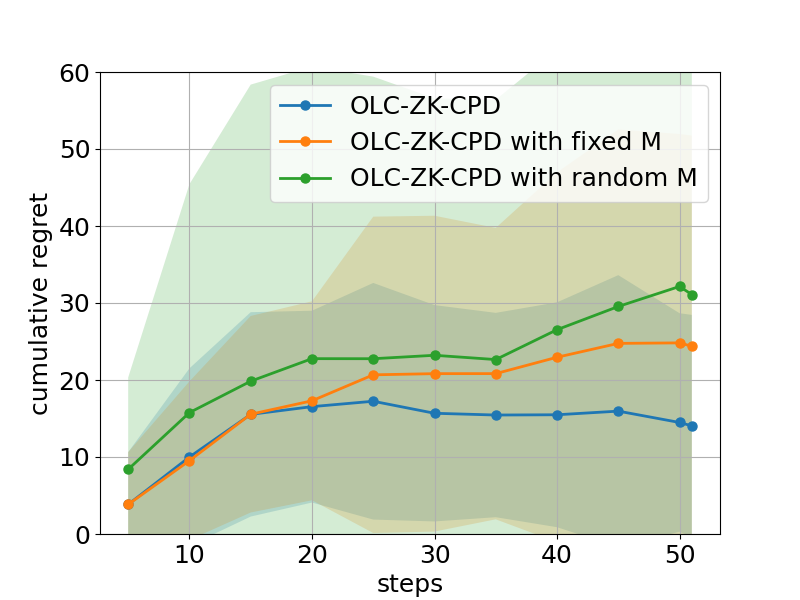}
  \label{POZK}}
\subfigure[]{
  \centering
  \includegraphics[scale=.2]{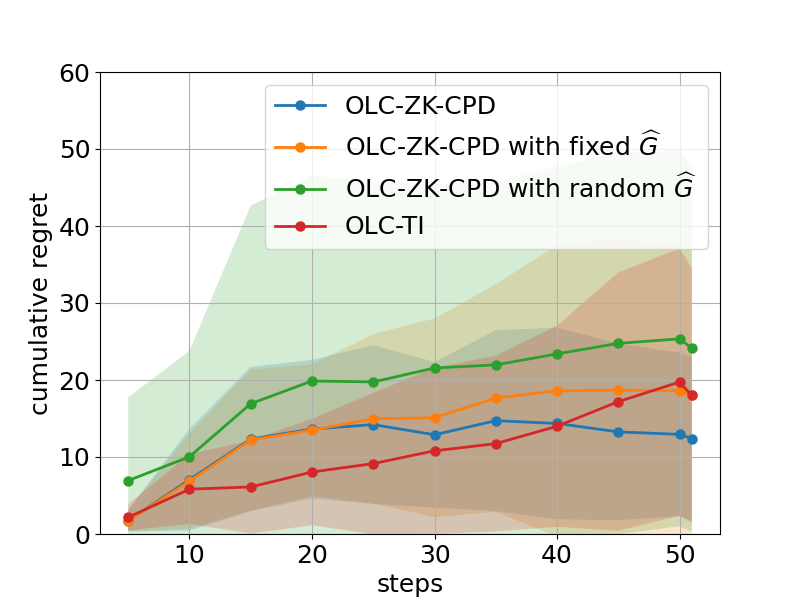}
  \label{different_G}}
\caption{(a) Cumulative regret of OLC-ZK-CPD with different $M$ estimation; (b) Cumulative regret of OLC-ZK-CPD with different $\widehat{G}$ estimation.}
\label{fig:std-0}
\end{figure}

\begin{figure}[h]
\centering
\subfigure[$h=2,N=4$]{%
  \centering
  \includegraphics[scale=.2]{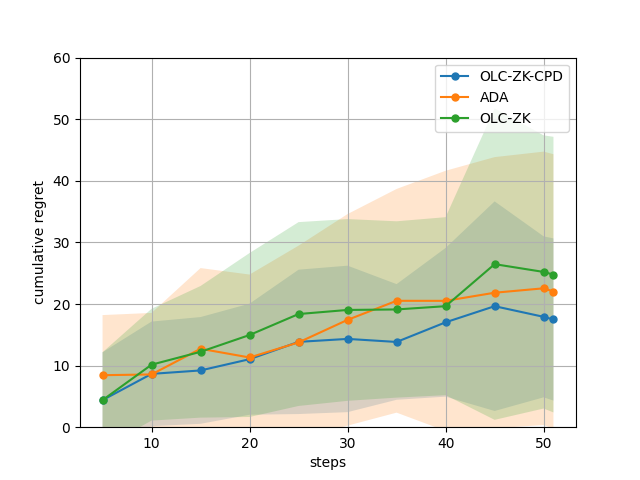}
	\label{h2N4}}
\subfigure[$h=2,N=6$.]{%
  \centering
  \includegraphics[scale=.2]{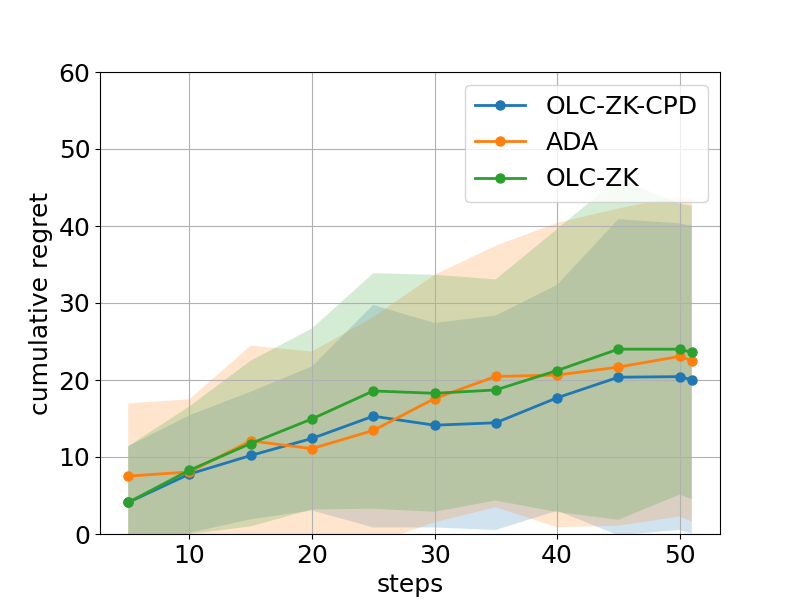}
	\label{h2N6}}
  \subfigure[$h=4,N=4$.]{%
  \centering
  \includegraphics[scale=.2]{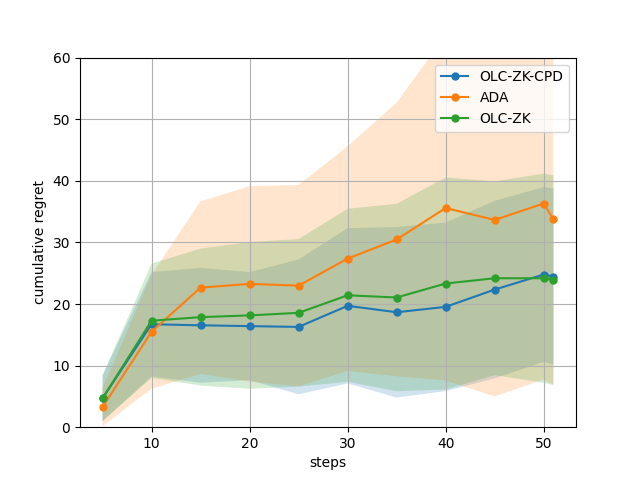}
	\label{h4N4}}
 \subfigure[$h=4,N=6$.]{%
  \centering
  \includegraphics[scale=.2]{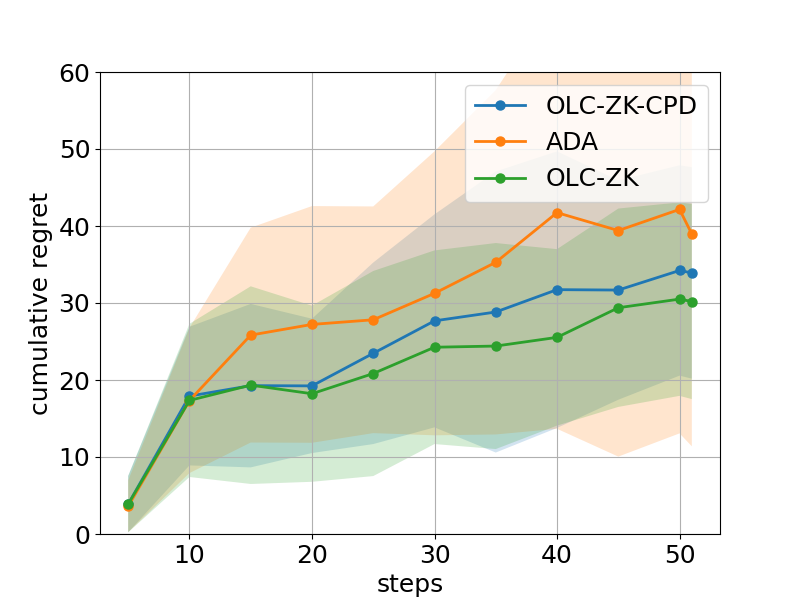}
	\label{h4N6}}
\caption{Performance Comparison with Baseline Algorithms for Time Varying Dynamical Systems.}
\label{fig:std}
\end{figure}

\section{Numerical Experiments}

In this section, experimental results are presented for illustrating the performance of OLC-ZK-CPD. 

\textbf{Parameter setting:} For all experiments, $\theta_t:=[A_t,B_t]$ and $w_t$ are randomly generated at each time step; $C_t$ is randomly initialized, but is kept unchanged across all time steps: $C_{t_1} = C_{t_2}, ~ \forall ~ t_1, t_2 \in[1,t]$, and $e_t=0, ~ \forall t \in [1,t]$.
The cost function is a quadratic function of $y_t$ and $u_t$: $c_t(y_t,u_t) = y_t^T Q y_t + u_t^T R u_t$. The matrices $Q$ and $R$ are randomly generated positive semi-definite matrices. Experiments are averaged over 10 random runs. In each run, all the algorithms use the same $Q, R, C, A_t$ and $B_t$. 

\textbf{Baselines}: Below, we describe the baseline algorithms we compare OLC-ZK-CPD with.
\begin{itemize}
\item {\it OLC-ZK}: is the online learning algorithm 
where the output $\widehat{G}^{cd}_k$ is itself used as the estimate of the system parameters for the duration of the period of the next interval of the change point detection procedure. At the end of the next interval, the estimate is updated to $\widehat{G}^{cd}_{k+1}$ and so on. 
\item {\it Adaptive Estimation Algorithm (ADA)}: is Algorithm \ref{alg:olc-zk-cd} with the estimation algorithm in \cite{minasyan2021online}, in place of the estimation approach in Algorithm \ref{alg:olc-zk-cd}. Essentially, in this combination, what is retained is only the policy parameter update step, with the entire estimation approach replaced by the adaptive estimation algorithm in \cite{minasyan2021online}. 
\item {\it OLC-TI}: is the online learning algorithm for time invariant systems \cite{simchowitz2020improper}. In contrast to ours, which continuously explores and exploits, OLC-TI explores first and then exploits. 
\item {\it OLC-ZK-CPD with fixed $M$}: is the online algorithm where $M$ is a fixed value and is not udpated. OLC-ZK-CPD with random $M$: the online algorithm where $M$ is picked randomly. 
\item {\it OLC-ZK-CPD with fixed $\widehat{G}$}: is the OLC-ZK-CPD algorithm with $\widehat{G}_t$ fixed to a constant value instead of an estimator. 
\item {\it OLC-ZK-CPD with random $\widehat{G}$}: is the OLC-ZK-CPD algorithm with $\widehat{G}_t$ picked randomly instead of an estimator. 
\end{itemize}

\textbf{Results:} In the figures, the shaded regions represent the standard deviation for the respective algorithms. \Cref{POZK} indicates that OLC-ZK-CPD has a smaller sub-linear increase in cumulative regret and smaller variance compared to the case when a fixed $M$ or a randomly generated $M$ is used instead.  
Similarly, it can be observed from \Cref{different_G} that the proposed OLC-ZK-CPD algorithm achieves a smaller sub-linear regret with smaller variance compared to the case when a fixed $\widehat{G}$ or a randomly generated $\widehat{G}$ is applied instead of Eq. \eqref{eq:ls-est}. Most importantly, while, initially the OLC-TI algorithm is better, over time its performance worsens and converges to the OLC-ZK-CPD with an arbitrarily fixed $\widehat{G}$. This is expected as the estimate from the initial exploration phase of OLC-TI can be very different from the underlying dynamical system after a sufficiently long time and thus behave like an arbitrarily fixed $\widehat{G}$ over time. These results corroborate the effectiveness of our proposed algorithm in adapting to time variations. 

In Fig. \ref{fig:std}, we compare the performance of OLC-ZK-CPD algorithm with the other adaptive algorithms for time-varying dynamical systems such as the OLC-ZK algorithm \cite{muthirayan2022adaptive} and the ADA algorithm. The plots are averaged over 10 random runs with parameters $N=[4,6]$ for $h=2$ and $N=[4,6]$ for $h=4$. In each run, all the algorithms are simulated with the same $Q, R, C, A_t$,  and $B_t$. The performance of ADA is unchanged with $N$ because it does not use the parameter $N$. 
We recall that $h$ defines the length of the history of disturbances in the DAC policy. For $h = 2$, we find that OLC-ZK-CPD achieves a better regret compared to ADA and OLC-ZK. For $h = 4$, we find that OLC-ZK-CPD achieves much better regret compared to ADA, and is similar to OLC-ZK, given that in one run, it is better compared to OLC-ZK, and in another run, it is worse. 
Overall, we also find that the statistical deviation of OLC-ZK-CPD is lesser compared to OLC-ZK and ADA, showing that OLC-ZK-CPD is more stable in the statistical sense. 






%% file: Conclusion.tex
\section{Conclusion}

In this work, we study the problem of online control of unknown time varying dynamical systems with arbitrary disturbances and cost functions. Our goal is to design an online adaptation algorithm that can provably achieve sub-linear regret upto sub-linear variations in the system with respect to stronger notions of variability like the number of changes. We present system, information and cost structures along with algorithms which guarantee such results and also present some open questions. 


%% file: Appendix.tex
\onecolumn

 \begin{definition}
\[ \widetilde{D} := \max\left\{\kappa_M\kappa_wh +\frac{\kappa_a\kappa_w}{\gamma} + \kappa_e + \frac{\kappa_a \kappa_b \kappa_M\kappa_wh}{\gamma},1\right\}, ~ L_f := L\widetilde{D}\left( \kappa_a \kappa_b \kappa_w\sqrt{h} + \kappa_w\sqrt{h}\right) \]
\[ G_f := G\widetilde{D}hnm\left(\frac{\kappa_a\kappa_b \kappa_w }{\gamma}+ \kappa_w\right), ~ D := \sup_{M_1, M_2 \in \mathcal{M}} \norm{M_1 - M_2}.  \]
\end{definition}

\section{Proof of Theorem \ref{thm:olc-fk}}
\label{sec:th1-proof}

\begin{remark}
When the cost functions are linear, it follows that \\
(i) For a given $z^\top = [y^\top, u^\top], (z')^\top = [(y')^\top, (u')^\top]$, \[\norm{c_t(y,u) -  c_t(y',u')} \leq L\norm{z - z'},\] where $L = G$.\\
(ii) $\nabla_y c(y,u) \leq G, \nabla_u c(y,u) \leq G$.
\end{remark}

Since $u^{\pi_{\mathrm{DAC-O}}}_t = \sum_{k=1}^m M^{[k]}_t w_{t-k+1}$, lets define a function $f_t$ by
\[ f_t(M_{t:t-h}\vert G_t , s_{1:t}) := c_t(\tilde{y}^{\pi_{\mathrm{DAC-O}}}_t[M_{t:t-h}\vert G_t , s_{1:t}],u^{\pi_{\mathrm{DAC-O}}}_t). \]

We first prove the following two intermediate results. 
\begin{lemma}
Consider two policy sequences $(M_{t-h}, \dots, M_{t-k}, \dots, M_t)$ and $(M_{t-h}, \dots, \widetilde{M}_{t-k}, \dots, M_t)$, which differ only in the policy at time $t-k$, where $k \in \{0,1,\dots,h\}$. Then,
\begin{align}
& \left\vert f_t(M_{t-h} \dots M_{t-k} \dots M_t\vert G_t , s_{1:t}) -  f_t(M_{t-h} \dots \widetilde{M}_{t-k} \dots M_t\vert G_t , s_{1:t})\right \vert \nonumber \\
& \leq L_f \norm{M_{t-k} - \widetilde{M}_{t-k}}, \nonumber \\
& L_f = L\max\left\{\kappa_M\kappa_wh + \frac{\kappa_a\kappa_w}{\gamma} +  \kappa_e +\frac{\kappa_a \kappa_b \kappa_M\kappa_wh}{\gamma},1\right\}\left( \kappa_a \kappa_b \kappa_w\sqrt{h} + \kappa_w\sqrt{h}\right). \nonumber  
\end{align}
\label{lem:lipschitzbound}
\end{lemma}

{\em Proof}:

Let $\widetilde{M}^k_{t:t-h} = [M_{t-h} \dots \widetilde{M}_{t-k} \dots M_t]$. Then 
\begin{align}
& \norm{\tilde{y}^{\pi_{\mathrm{DAC-O}}}_t[M_{t:t-h}\vert G_t , s_{1:t}] -   \tilde{y}^{\pi_{\mathrm{DAC-O}}}_t[\widetilde{M}^k_{t:t-h} \vert G_t , s_{1:t}]} = \norm{G^{[k]}_t\left( \sum_{i=1}^h \left( M^{[i]}_{t-k} -  \widetilde{M}^{[i]}_{t-k}\right) w_{t-k-i+1}\right) } \nonumber \\
& \stackrel{(a)}{\leq} \norm{G^{[k]}_t}_2 \left( \sum_{i=1}^h \norm{M^{[i]}_{t-k} -  \widetilde{M}^{[i]}_{t-k}} \norm{w_{t-k-i+1}}\right) \stackrel{(b)}{\leq} (1-\gamma)^{k-1}\kappa_a \kappa_b \left( \sum_{i=1}^h \norm{M^{[i]}_{t-k} -  \widetilde{M}^{[i]}_{t-k}} \norm{w_{t-k-i+1}} \right) \nonumber \\
& \stackrel{(c)}{\leq} (1-\gamma)^{k-1}\kappa_a \kappa_b \kappa_w \left( \sum_{i=1}^h \norm{M^{[i]}_{t-k} -  \widetilde{M}^{[i]}_{t-k}}  \right) \stackrel{(d)}{\leq} (1-\gamma)^{k-1}\kappa_a \kappa_b \kappa_w\sqrt{h} \norm{M_{t-k} -  \widetilde{M}_{t-k}}. \nonumber 
\end{align}

Here, $(a)$ follows from using the definition of $2$-norm of a matrix and triangle inequality, $(b, c)$ follow from Assumption \ref{ass:sys} and $(d)$ follows from the inequality that for any positive numbers $a_1, a_2, \dots, a_h$, $\sum a_k \leq \sqrt{h} \sqrt{\sum a_k^2}$. Let 
\[ \tilde{u}^{\pi_{\mathrm{DAC-O}}}_t = \left\{ \begin{array}{cc} \sum_{i = 1}^h M^{[i]}_t w_{t-i+1} & \tn{If} ~ k > 0 \\ \sum_{i = 1}^h \widetilde{M}^{[i]}_t w_{t-i+1} &  \tn{If} ~ k = 0. 
 \end{array} \right. \] 

Then,
\[
\norm{u^{\pi_{\mathrm{DAC-O}}}_t - \tilde{u}^{\pi_{\mathrm{DAC-O}}}_t} \stackrel{(e)}{\leq} \sum_{i=1}^h \norm{M^{[i]}_{t} -  \widetilde{M}^{[i]}_{t}}\norm{w_{t-i+1}} \stackrel{(f)}{\leq} \kappa_w\sqrt{h}\norm{M_{t-k} -  \widetilde{M}_{t-k}}.  \]

Here, $(e)$ follows from applying triangle inequality first and by using the definition of $2$-norm of a matrix, $(f)$ follows from Assumption \ref{ass:sys}. We make the following observations. Using a similar argument it follows that the control input $u^{\pi_{\mathrm{DAC-O}}}_t$ is bounded by
\[ \norm{u^{\pi_{\mathrm{DAC-O}}}_t} \leq \kappa_w\kappa_Mh.\]

The truncated output $\tilde{y}^{\pi_{\mathrm{DAC-O}}}_t[M_{t:t-h}\vert G_t , s_{1:t}]$ is bounded by 
\begin{align}
& \norm{\tilde{y}^{\pi_{\mathrm{DAC-O}}}_t[M_{t:t-h}\vert G_t , s_{1:t}]} \stackrel{(g)}{\leq} \norm{s_t}  + \sum_{k = 1}^{h} \norm{G^{[k]}_t}_2 \norm{u^{\pi_{\mathrm{DAC-O}}}_{t-k}} \nonumber \\
& \stackrel{(h)}{\leq} \norm{s_t}  + \kappa_a \kappa_b \kappa_M\kappa_wh \sum_{k = 1}^{h}  (1-\gamma)^{k-1} \stackrel{(i)}{\leq} \frac{\kappa_a\kappa_w}{\gamma} + \kappa_e + \frac{\kappa_a \kappa_b \kappa_M\kappa_wh}{\gamma}. \nonumber 
\end{align}

Here, $(g)$ follows from triangle inequality, $(h)$ follows from Assumption \ref{ass:sys} and the bound on $\norm{u^{\pi_{\mathrm{DAC-O}}}_t}$ and $(i)$ follows from bounding $s_t$ using the definition of $s_t$ and Assumption \ref{ass:sys}. Given these observations, by Assumption \ref{ass:lipschitz} it follows that
\begin{align}
& \left\vert f_t(M_{t-h} \dots M_{t-k} \dots M_t\vert G_t , s_{1:t}) -  f_t(M_{t-h} \dots \widetilde{M}_{t-k} \dots M_t\vert G_t , s_{1:t})\right \vert \nonumber \\
& \leq LR\left( \norm{\tilde{y}^{\pi_{\mathrm{DAC-O}}}_t[M_{t:t-h}\vert G_t , s_{1:t}] - \tilde{y}^{\pi_{\mathrm{DAC-O}}}_t[\widetilde{M}^k_{t:t-h} \vert G_t , s_{1:t}]}\right) + LR \left(\norm{u^{\pi_{\mathrm{DAC-O}}}_t - \tilde{u}^{\pi_{\mathrm{DAC-O}}}_t}\right) \nonumber \\
& \stackrel{(j)}{\leq} LR\left( (1-\gamma)^{k-1}\kappa_a \kappa_b \kappa_w\sqrt{h} + \kappa_w\sqrt{h}\right) \norm{M_{t-k} -  \widetilde{M}_{t-k}} \nonumber \\
& \stackrel{(h)}{\leq} L\max\left\{\kappa_M\kappa_wh + \frac{\kappa_a\kappa_w}{\gamma} + \kappa_e + \frac{\kappa_a \kappa_b \kappa_M\kappa_wh}{\gamma},1\right\}\left( (1-\gamma)^{k-1}\kappa_a \kappa_b \kappa_w\sqrt{h} + \kappa_w\sqrt{h}\right) \norm{M_{t-k} -  \widetilde{M}_{t-k}}. \nonumber 
\end{align}

Here, $(j)$ follows from using the fact that the arguments $\tilde{y}^{\pi_{\mathrm{DAC-O}}}_t, u^{\pi_{\mathrm{DAC-O}}}_t$ and $\tilde{u}^{\pi_{\mathrm{DAC-O}}}_t$ are bounded as shown earlier and Assumption \ref{ass:lipschitz}, and $(h)$ follows from substituting for $R$ using the bounds on $\tilde{y}^{\pi_{\mathrm{DAC-O}}}_t$ and $u^{\pi_{\mathrm{DAC-O}}}_t$ derived earlier. The final result follows from here $\square$

\begin{lemma}
For all $M$ such that $\norm{M^{[j]}} \leq \kappa_M, ~ \forall ~ j \in [1,h]$, we have that
\begin{align}
& \norm{\nabla_M f_t(M,\dots,M\vert G_t , s_{1:t})} \leq G_f, ~~ G_f = GDhnm\left(\frac{\kappa_a\kappa_b \kappa_w }{\gamma}+ \kappa_w\right),  \nonumber \\
& D = \max\left\{\kappa_M\kappa_wh + \frac{\kappa_a\kappa_w}{\gamma} + \kappa_e + \frac{\kappa_a \kappa_b \kappa_M\kappa_wh}{\gamma}, 1\right\}. \nonumber 
\end{align}
\label{lem:gradbound}
\end{lemma}

{\em Proof}:
When each $M_k \in M_{t:t-h}$ is fixed to $M$, we denote $f_t(M_{t:t-h}\vert G_t , s_{1:t})$ by $F_t(M\vert G_t , s_{1:t})$. Thus,
\[ \nabla_M f_t(M,\dots,M\vert G_t , s_{1:t}) = \nabla_M F_t(M\vert G_t , s_{1:t}).\]

Let 
\[D = \max\left\{\kappa_M\kappa_wh + \frac{\kappa_a\kappa_w}{\gamma} + \kappa_e + \frac{\kappa_a \kappa_b \kappa_M\kappa_wh}{\gamma}, 1\right\} \]

Let $\pi_{\mathrm{DAC}}$ denote the fixed DAC policy with the policy parameter $M \in \mathcal{M}$. Then, the output and the control input under $\pi_{\mathrm{DAC}}$ are given by $y^{\pi_{\mathrm{DAC}}}_t$ and $u^{\pi_{\mathrm{DAC}}}_t$. Then, similar to the derivation in Lemma \ref{lem:lipschitzbound} $\norm{\tilde{y}^{\pi_{\mathrm{DAC}}}_t} \leq D, \norm{u^{\pi_{\mathrm{DAC}}}_t} \leq D$. Let $M^{[k]}_{p,q}$ denote the $(p,q)$th element of the matrix $M^{[k]}$. Then, it is sufficient to derive the bound for $\nabla_{M^{[k]}_{p,q}} F_t(M\vert G_t , s_{1:t})$, to bound the overall gradient. Then, from Assumption \ref{ass:lipschitz}.{\it (iii)} it follows that
\[ \norm{\nabla_{M^{[k]}_{p,q}} F_t(M \vert G_t , s_{1:t})} \leq GD\left(\norm{\frac{\partial \tilde{y}^{\pi_{\mathrm{DAC}}}_t}{\partial M^{[k]}_{p,q}}} + \norm{\frac{\partial u^{\pi_{\mathrm{DAC}}}_t}{\partial M^{[k]}_{p,q}}}\right). \]

We have that
\[\frac{\partial \tilde{y}^{\pi_{\mathrm{DAC}}}_t}{\partial M^{[k]}_{p,q}} = \sum_{i = 1}^h G^{[i]}_t  \frac{\partial u^{\pi_{\mathrm{DAC}}}_{t-i}}{\partial M^{[k]}_{p,q}} = \sum_{i = 1}^h G^{[i]}_t  \left( \sum_{j = 1}^h \frac{\partial M}{\partial M^{[j]}_{p,q}} \right) w_{t-i-j+1} \]

Then taking the norm on both sides, 
\begin{align}
& \norm{\frac{\partial \tilde{y}^{\pi_{\mathrm{DAC}}}_t}{\partial M^{[k]}_{p,q}}} \stackrel{(a)}{\leq} \sum_{i=1}^h \sum_{j=1}^h \norm{G^{[i]}_t}_2 \norm{\frac{\partial M^{[j]}}{\partial M^{[k]}_{p,q}}}\norm{w_{t-i-j+1}} \stackrel{(b)}{=} \sum_{i=1}^h \norm{G^{[i]}_t}_2\norm{w_{t-i-k+1}}  \nonumber \\
& \stackrel{(c)}{\leq} \sum_{i=1}^h \left( 1 -\gamma\right)^{i-1}\kappa_a\kappa_b \kappa_w \stackrel{(d)}{\leq} \frac{\kappa_a\kappa_b \kappa_w }{\gamma}. \nonumber  
\end{align}

Here, $(a)$ follows from applying triangle inequality and by using the definition of $2$-norm of a matrix, $(b)$ follows from the fact that $\norm{\frac{\partial M^{[j]}}{\partial M^{[k]}_{p,q}}} = 1$ for $j = k$ and zero otherwise, and $(c)$ follows from Assumption \ref{ass:sys}. Similarly, we have that
\[\frac{\partial u^{\pi_{\mathrm{DAC}}}_t}{\partial M^{[k]}_{p,q}} \stackrel{(f)}{=} \sum_{i = 1}^{h} \frac{\partial M^{[i]}}{\partial M^{[k]}_{p,q}} w_{t-i+1}, ~~ \tn{i.e.,} ~~ \norm{\frac{\partial u^{\pi_{\mathrm{DAC}}}_t}{\partial M^{[k]}_{p,q}} } \stackrel{(g)}{\leq} \norm{w_{t-k+1}} = \kappa_w. \]

Here, $(f)$ follows from the definition of $u^{\pi_{\mathrm{DAC}}}_t$ and $(g)$ follows from the fact that $\norm{\frac{\partial M^{[i]}}{\partial M^{[k]}_{p,q}}} = 1$ for $i = k$ and zero otherwise. Thus, we get
\[ \norm{\nabla_{M^{[k]}_{p,q}} F_t(M \vert G_t , s_{1:t})} \leq GD\left(\frac{\kappa_a\kappa_b \kappa_w }{\gamma}+ \kappa_w\right).\] 

The final result follows from here $\square$

We now prove the main result. Let $M^{\star}_{1:T}$ be the best switching DAC policy. Let $\Pi_\mathcal{M}$ denote the class of switching DAC policies. 
Given this definition, the regret can be split as
\begin{align}
& R_T = \sum_{t=1}^T c_t(y^{\pi_{\mathrm{DAC-O}}}_t, u^{\pi_{\mathrm{DAC-O}}}_t) - \min_{\pi \in \Pi_\mathcal{M}}\sum_{t=1}^T c_t(y^{\pi}_t, u^{\pi}_t) \nonumber \\
& = \underbrace{\sum_{t=1}^T c_t(y^{\pi_{\mathrm{DAC-O}}}_t, u^{\pi_{\mathrm{DAC-O}}}_t) - \sum_{t=1}^T c_t(\tilde{y}^{\pi_{\mathrm{DAC-O}}}_t[M_{t:t-h}\vert G_t , s_{1:t}],u^{\pi_{\mathrm{DAC-O}}}_t)}_{\tn{Cost Truncation Error}} \nonumber \\
& + \underbrace{\sum_{t=1}^T f_t(M_{t:t-h}\vert G_t , s_{1:t}) - \sum_{t=1}^T f_t(M^{\star}_{t:t-h}\vert G_t , s_{1:t})}_{\tn{Policy Regret}} \nonumber \\
& + \underbrace{\sum_{t=1}^T f_t(M^{\star}_{t:t-h}\vert G_t , s_{1:t}) - \min_{\pi \in \Pi_\mathcal{M}}\sum_{t=1}^T c_t(y^{\pi}_t, u^{\pi}_t)}_{\tn{Cost Approximation Error}}. \nonumber 
\end{align}

Here, the cost truncation error is the error due to replacing the full output $y^{\pi_{\mathrm{DAC-O}}}_t$ with the truncated output $\tilde{y}^{\pi_{\mathrm{DAC-O}}}_t[M_{t:t-h}\vert G_t , s_{1:t}]$, the cost approximation error is the error of the optimal given the truncation of the cost, and the policy approximation error is the error between the optimal DAC and and the optimal linear feedback policy. 

Let $D = \sup_{M_1, M_2 \in \mathcal{M}} \norm{M_1 - M_2}$. By applying \cite[Theorem 2]{zhao2022non} to the online learning algorithm Algorithm \ref{alg:olc-fk}, and Lemma \ref{lem:lipschitzbound} and Lemma \ref{lem:gradbound}, we get that
\[ \sum_{t=1}^T f_t(M_{t:t-h}\vert G_t , s_{1:t}) - \sum_{t=1}^T f_t(M^{\star}_{t:t-h}\vert G_t , s_{1:t}) \leq \mathcal{O}(\sqrt{T(1+P_T)}. \] 

Next, we bound the cost truncation error term. 
\begin{align}
& \sum_{t=1}^T c_t(y^{\pi_{\mathrm{DAC-O}}}_t, u^{\pi_{\mathrm{DAC-O}}}_t) - \sum_{t=1}^T c_t(\tilde{y}^{\pi_{\mathrm{DAC-O}}}_t[M_{t:t-h}\vert G_t , s_{1:t}],u^{\pi_{\mathrm{DAC-O}}}_t) \nonumber\\
& \stackrel{(a)}{\leq} \sum_{t=1}^T LR\norm{y^{\pi_{\mathrm{DAC-O}}}_t-\tilde{y}^{\pi_{\mathrm{DAC-O}}}_t[M_{t:t-h}\vert G_t , s_{1:t}]} \stackrel{(b)}{\leq} \sum_{t=1}^T LR \norm{ \sum_{k = h+1}^{t-1} G^{[k]}_tu^{\pi_{\mathrm{DAC-O}}}_{t-k}} \nonumber \\
& \stackrel{(c)}{\leq} \sum_{t=1}^T LR \sum_{k = h+1}^{t-1} \norm{G^{[k]}_t}_2\norm{u^{\pi_{\mathrm{DAC-O}}}_{t-k}} \stackrel{(d)}{\leq} \sum_{t=1}^T LR \kappa_a \kappa_b \kappa_M \kappa_wh \sum_{k = h+1}^{t-1} (1 - \gamma)^{k-1} \leq \frac{LR \kappa_a \kappa_b \kappa_M \kappa_w h (1-\gamma)^h  T}{\gamma}. \nonumber 
\end{align}

Here, $(a)$ follows from Assumption \ref{ass:lipschitz}, $(b)$ follows from definitions of $y^{\pi_{\mathrm{DAC-O}}}_t$ and $\tilde{y}^{\pi_{\mathrm{DAC-O}}}_t[M_{t:t-h}\vert G_t , s_{1:t}]$, $(c)$ follows from applying triangle inequality and by using the definition of $2$-norm of a matrix and $(d)$ follows from Assumption \ref{ass:sys} and the bound on $u^{\pi_{\mathrm{DAC-O}}}_t$.

Next, we bound the cost approximation error term. This term can also be bounded in a similar way. 
Then
\begin{align}
&  \sum_{t = 1}^T f_t(M^{\star}_{t:t-h} \vert G_t , s_{1:t}) - \min_{\pi \in \Pi_\mathcal{M}}\sum_{t=1}^T c_t(y^{\pi}_t, u^{\pi}_t) \leq \sum_{t = 1}^T f_t(M^{\star}_{t:t-h} \vert G_t , s_{1:t}) - \min_{\pi \in \Pi_\mathcal{M}}\sum_{t=1}^T c_t(y^{\pi}_t, u^{\pi}_t) \nonumber \\
& \stackrel{(e)}{\leq}  LR \sum_{t = 1}^T \left( \norm{\tilde{y}^{\pi_{\mathrm{DAC-O}}}_t[M^{\star}_{t:t-h}\vert G_t , s_{1:t}] - y^{\pi}_t} \right) \stackrel{(f)}{\leq} LR \sum_{t = 1}^T \sum_{k = h+1}^{t-1} \norm{G^{[k]}_{t}}_2 \norm{u^\pi_{t-k}} \stackrel{(g)}{\leq} \frac{LR \kappa_a\kappa_b\kappa_M\kappa_w h (1-\gamma)^{h}T}{\gamma}. \nonumber 
\end{align}

Here, $(e)$ follows from Assumption \ref{ass:lipschitz}, $(f)$ follows from the fact that $\tilde{y}^{\pi_{\mathrm{DAC-O}}}_t[M^{\star}_{t:t-h}\vert G_t , s_{1:t}]$ is just the truncated output of $y^{\pi}_t$ and applying triangle followed by using the definition of $2$-norm of a matrix, and $(g)$ follows from applying Assumption \ref{ass:sys}. 

Putting together all the three terms, we get the final result $\blacksquare$

\section{Proof of Theorem \ref{thm:olc-zk-cd} : Setting S-2 }
\label{sec:th2-proof}

\begin{lemma}
Under Algorithm \ref{alg:olc-zk-cd}, with a probability $1-\delta/3$, $\delta \leq 1/T$ and $T \geq 3$
\[\norm{u^{\pi_{\mathrm{DAC-O}}}_t} \leq R_u = \kappa_M\kappa_wh + 3\sigma\sqrt{m+\log(1/\delta)}, ~~ \norm{\hat{s}_t} \leq R_s = \frac{\kappa_a\kappa_w}{\gamma} + \kappa_e + \frac{2R_u\kappa_a\kappa_b}{\gamma},\]\[ \norm{y^{\pi_{\mathrm{DAC-O}}}_t} \leq \frac{\kappa_a\kappa_w}{\gamma} + \kappa_e + \frac{\kappa_a\kappa_bR_u}{\gamma} \leq R_s, ~~ \forall ~~ t.\]
\end{lemma}

{\em Proof}:
The first term in the bound of $\norm{u^{\pi_{\mathrm{DAC-O}}}_t}$ follows from bounding the DAC part of the control policy \ref{eq:cp-zk-o} following the steps in the proof of Lemma \ref{lem:lipschitzbound}. The second term follows from bounding the perturbation Eq. \eqref{eq:pert-zk-o} by using \cite[Claim D.3]{simchowitz2020improper}.

From the definition of $\hat{s}_t$, we get under the event that $\norm{u^{\pi_{\mathrm{DAC-O}}}_t} \leq R_u$,
\begin{align}
    & \hat{s}_t = s_t + \sum_{k=1}^h \left( G^{[k]}_t - \widehat{G}^{[k]}_t\right)u^{\pi_{\mathrm{DAC-O}}}_t, ~~ \tn{i.e.,} ~~ \norm{\hat{s}_t} \stackrel{(a)}{\leq} \norm{s_t} + \sum_{k=1}^h \norm{G^{[k]}_t - \widehat{G}^{[k]}_t}_2 \norm{u^{\pi_{\mathrm{DAC-O}}}_t} \nonumber \\
    & \norm{\hat{s}_t} \stackrel{(b)}{\leq} \frac{\kappa_a\kappa_w}{\gamma} + \kappa_e + 2R_u\sum_{k=1}^h \left(1-\gamma\right)^{k-1}\kappa_a\kappa_b = \frac{\kappa_a\kappa_w}{\gamma} + \kappa_e + \frac{2R_u\kappa_a\kappa_b}{\gamma}. \nonumber 
\end{align} 

Here, $(a)$ follows from applying triangle inequality, followed by the definition of $2$-norm of a matrix and $(b)$ follows from Assumption \ref{ass:sys}. By definition, we get, under the event that $\norm{u^{\pi_{\mathrm{DAC-O}}}_t} \leq R_u$,
\begin{align}
& y^{\pi_{\mathrm{DAC-O}}}_t = s_t + \sum_{k = 1}^{t-1} G^{[k]}_tu^{\pi_{\mathrm{DAC-O}}}_{t-k}, ~~ \tn{i.e.}, ~~ \norm{y^{\pi_{\mathrm{DAC-O}}}_t} \leq \norm{s_t} + \sum_{k = 1}^{t-1} \norm{G^{[k]}_t}_2\norm{u^{\pi_{\mathrm{DAC-O}}}_{t-k}} \nonumber \\
& \leq \frac{\kappa_a\kappa_w}{\gamma} + \kappa_e + \sum_{k = 1}^{t-1} (1 -\gamma)^{k-1}\kappa_a\kappa_bR_u = \frac{\kappa_a\kappa_w}{\gamma} + \kappa_e + \frac{\kappa_a\kappa_bR_u}{\gamma} \square\nonumber 
\end{align}

Lets define a function $f_t$ by
\[ f_t(M_{t:t-h}\vert \widehat{G}_t , \hat{s}_{1:t}) := c_t(\tilde{\tilde{y}}^{\pi_{\mathrm{DAC-O}}}_t[M_{t:t-h}\vert \widehat{G}_t , \hat{s}_{1:t}],\tilde{u}^{\pi_{\mathrm{DAC-O}}}_t[M_{t} \vert w_{1:t}]). \]

\begin{lemma}
Consider two policy sequences $(M_{t-h}, \dots, M_{t-k}, \dots, M_t)$ and $(M_{t-h}, \dots, \widetilde{M}_{t-k}, \dots, M_t)$, which differ only in the policy at time $t-k$, where $k \in \{0,1,\dots,h\}$. Then, with probability $1-\delta/3$
\begin{align}
& \left\vert f_t(M_{t-h} \dots M_{t-k} \dots M_t\vert \widehat{G}_t , \hat{s}_{1:t}) -  f_t(M_{t-h} \dots \widetilde{M}_{t-k} \dots M_t\vert \widehat{G}_t , \hat{s}_{1:t})\right \vert \nonumber \\
& \leq L_f \norm{M_{t-k} - \widetilde{M}_{t-k}}, ~~ L_f =  L\max\left\{\left(1+\frac{\kappa_a\kappa_b}{\gamma}\right)R_u+R_s,1\right\}\left((1-\gamma)^{k-1}\kappa_a\kappa_b\kappa_w\sqrt{h} + \kappa_w\sqrt{h}\right). \nonumber  
\end{align}
\label{lem:lipschitzbound-zk}
\end{lemma}

{\em Proof}:
The proof follows from the same steps as in the proof of Lemma \ref{lem:lipschitzbound} $\square$

\begin{lemma}
For all $M$ such that $\norm{M^{[j]}} \leq \kappa_M, ~ \forall ~ j \in [1,h]$, we have that
\[ \norm{\nabla_M f_t(M,\dots,M\vert \widehat{G}_t , \hat{s}_{1:t})} \leq G_f, ~~ G_f = GD\left(\frac{\kappa_a\kappa_b\kappa_w}{\gamma}+\kappa_w\right), ~~ D = \max\left\{\left(1+\frac{\kappa_a\kappa_b}{\gamma}\right)R_u+R_s,1\right\}. \]
\label{lem:gradbound-zk}
\end{lemma}

{\em Proof}:
The proof follows from the same steps as in the proof of Lemma \ref{lem:gradbound} $\square$

We now prove the main result. Let $M^{\star}_{1:T}$ be the best switching DAC policy. Let $\Pi_\mathcal{M}$ denote the class of switching DAC policies. We can split the regret as
\begin{align}
& R_T = \sum_{t=1}^T c_t(y^{\pi_{\mathrm{DAC-O}}}_t, u^{\pi_{\mathrm{DAC-O}}}_t) - \min_{\pi \in \Pi_\mathcal{M}}\sum_{t=1}^T c_t(y^{\pi}_t, u^{\pi}_t) \nonumber \\
& \stackrel{(a)}{=} \underbrace{\sum_{t=1}^T c_t(y^{\pi_{\mathrm{DAC-O}}}_t, u^{\pi_{\mathrm{DAC-O}}}_t) - \sum_{t=1}^T c_t(\tilde{y}^{\pi_{\mathrm{DAC-O}}}_t[M_{t:t-h}\vert \widehat{G}_t , \hat{s}_{1:t}], u^{\pi_{\mathrm{DAC-O}}}_t)}_{\tn{Cost Truncation Error}} \nonumber \\
& + \underbrace{\sum_{t=1}^T c_t(\tilde{y}^{\pi_{\mathrm{DAC-O}}}_t[M_{t:t-h}\vert \widehat{G}_t , \hat{s}_{1:t}], u^{\pi_{\mathrm{DAC-O}}}_t) -  \sum_{t=1}^T f_t(M^{\star}_{t:t-h}\vert \widehat{G}_t , \hat{s}_{1:t})}_{\tn{Policy Regret}} \nonumber \\
& + \underbrace{\sum_{t=1}^T f_t(M^{\star}_{t:t-h}\vert \widehat{G}_t , \hat{s}_{1:t}) - \sum_{t=1}^T f_t(M^{\star}_{t:t-h}\vert G_t , s_{1:t})}_{\tn{Model Approximation Error}} \nonumber \\
& + \underbrace{ \sum_{t=1}^T f_t(M^{\star}_{t:t-h}\vert G_t , s_{1:t}) - \min_{\pi \in \Pi_\mathcal{M}}\sum_{t=1}^T c_t(y^{\pi}_t, u^{\pi}_t)}_{\tn{Policy Approximation Error}}. \nonumber 
\end{align}

By definition 
\[ \tilde{y}^{\pi_{\mathrm{DAC-O}}}_t[M_{t:t-h}\vert \widehat{G}_t , \hat{s}_{1:t}] = \hat{s}_t + \sum_{k = 1}^h \widehat{G}^{[k]}_tu^{\pi_{\mathrm{DAC-O}}}_{t-k} = y^{\pi_{\mathrm{DAC-O}}}_t.\]

Thus,
\[ \underbrace{\sum_{t=1}^T c_t(y^{\pi_{\mathrm{DAC-O}}}_t, u^{\pi_{\mathrm{DAC-O}}}_t) - \sum_{t=1}^T c_t(\tilde{y}^{\pi_{\mathrm{DAC-O}}}_t[M_{t:t-h}\vert \widehat{G}_t , \hat{s}_{1:t}], u^{\pi_{\mathrm{DAC-O}}}_t)}_{\tn{Cost Truncation Error}} = 0.\]

Next, we bound the policy regret term. Let $R = \max\left\{\left(1+\frac{\kappa_a\kappa_b}{\gamma}\right)R_u+R_s,1\right\}$. Then, by using the definition of $u^{\pi_{\mathrm{DAC-O}}}_t$, Assumption \ref{ass:lipschitz} and \cite[Claim D.3]{simchowitz2020improper}, with a probability $1-\delta/3$
\[ \sum_{t=1}^T c_t(\tilde{y}^{\pi_{\mathrm{DAC-O}}}_t[M_{t:t-h}\vert \widehat{G}_t , \hat{s}_{1:t}], u^{\pi_{\mathrm{DAC-O}}}_t) \leq \sum_{t=1}^T \left( f_t(M_{t:t-h}\vert \widehat{G}_t , \hat{s}_{1:t}) + 3LR\left(1+\frac{\kappa_a\kappa_b}{\gamma}\right)\sigma\sqrt{m+\log(1/\delta)}\right) \]

Thus, with a probability $1-\delta/3$, the policy regret term is given by
\begin{align}
& \sum_{t=1}^T c_t(\tilde{y}^{\pi_{\mathrm{DAC-O}}}_t[M_{t:t-h}\vert \widehat{G}_t , \hat{s}_{1:t}], u^{\pi_{\mathrm{DAC-O}}}_t) - \sum_{t=1}^T f_t(M^{\star}_t\vert \widehat{G}_t , \hat{s}_{1:t}) \nonumber \\
& \leq \sum_{t=1}^T f_t(M_{t:t-h}\vert \widehat{G}_t , \hat{s}_{1:t}) - \sum_{t=1}^T f_t(M^{\star}_{t:t-h} \vert \widehat{G}_t , \hat{s}_{1:t}) + 3LR\left(1+\frac{\kappa_a\kappa_b}{\gamma}\right)\sqrt{m+\log(1/\delta)}T\sigma. \nonumber 
\end{align}

Let $D = \sup_{M_1, M_2 \in \mathcal{M}} \norm{M_1 - M_2}$. By applying \cite[Theorem 2]{zhao2022non} to the online learning algorithm Algorithm \ref{alg:olc-zk-cd}, and Lemma \ref{lem:lipschitzbound-zk} and Lemma \ref{lem:gradbound-zk}, we get that
\[ \sum_{t=1}^T f_t(M_{t:t-h}\vert \widehat{G}_t , \hat{s}_{1:t}) -  \sum_{t=1}^T f_t(M^\star_{t:t-h}\vert \widehat{G}_t , \hat{s}_{1:t}) \leq \mathcal{O}\left(\sqrt{T(1+P_T)}\right). \]

Thus, under an event with a probability $1-\delta/3$
\begin{align} 
& \sum_{t=1}^T c_t(\tilde{y}^{\pi_{\mathrm{DAC-O}}}_t[M_{t:t-h}\vert \widehat{G}_t , \hat{s}_{1:t}], u^{\pi_{\mathrm{DAC-O}}}_t) - \sum_{t=1}^T f_t(M^\star_{t:t-h}\vert \widehat{G}_t , \hat{s}_{1:t}) \nonumber \\
& \leq \mathcal{O}\left(\sqrt{T(1+P_T)}\right) + 3LR\left(1+\frac{\kappa_a\kappa_b}{\gamma}\right)\sqrt{m+\log(1/\delta)}T\sigma. \nonumber 
\end{align}

Next, we bound the model approximation term. 
\begin{align}
&  \sum_{t=1}^T f_t(M^\star_{t:t-h}\vert \widehat{G}_t , \hat{s}_{1:t}) -  \sum_{t=1}^T f_t(M^\star_{t:t-h}\vert G_t , s_{1:t}) \stackrel{(a)}{\leq} \nonumber \sum_{t=1}^T f_t(M^\star_{t:t-h}\vert \widehat{G}_t , \hat{s}_{1:t})  - \sum_{t=1}^T f_t(M^\star_{t:t-h}\vert G_t , s_{1:t}) \\
&  \stackrel{(b)}{=} \sum_{t=1}^T c_t(\tilde{\tilde{y}}^{\pi_{\mathrm{DAC-O}}}_t[M^\star_{t:t-h} \vert \widehat{G}_t , \hat{s}_{1:t}],\tilde{u}^{\pi_{\mathrm{DAC-O}}}_t[M^\star_t \vert w_{1:t}]) - \sum_{t=1}^T c_t(\tilde{\tilde{y}}^{\pi_{\mathrm{DAC-O}}}_t[M^\star_{t:t-h} \vert G_t , s_{1:t}],\tilde{u}^{\pi_{\mathrm{DAC-O}}}_t[M^\star_t \vert w_{1:t}]) \nonumber \\
& \stackrel{(c)}{\leq} LR\sum_{t=1}^T\left( \norm{\tilde{\tilde{y}}^{\pi_{\mathrm{DAC-O}}}_t[M^\star_{t:t-h} \vert \widehat{G}_t , \hat{s}_{1:t}] - \tilde{\tilde{y}}^{\pi_{\mathrm{DAC-O}}}_t[M^\star_{t:t-h} \vert G_t , s_{1:t}]}\right). \nonumber 
\end{align}

Here, 
$(b)$ follows from the respective definitions and $(c)$ follows from Assumption \ref{ass:lipschitz}. Now, under the same event, with probability $1-\delta/3$,
\begin{align}
& \norm{\tilde{\tilde{y}}^{\pi_{\mathrm{DAC-O}}}_t[M^\star_{t:t-h} \vert \widehat{G}_t , \hat{s}_{1:t}] - \tilde{\tilde{y}}^{\pi_{\mathrm{DAC-O}}}_t[M^\star_t \vert G_t , s_{1:t}]} \nonumber \\
& = \norm{ (\hat{s}_t - s_t) + \sum_{k=1}^h  \widehat{G}^{[k]}_t\tilde{u}^{\pi_{\mathrm{DAC-O}}}_{t-k}[M^\star_{t-k} \vert w_{1:t}] - \sum_{k=1}^h  G^{[k]}_t\tilde{u}^{\pi_{\mathrm{DAC-O}}}_{t-k}[M^\star_{t-k}\vert w_{1:t}] } \nonumber \\
& = \norm{ (\hat{s}_t - s_t) + \sum_{k=1}^h  \left(\widehat{G}^{[k]}_t - G^{[k]}_t\right)\tilde{u}^{\pi_{\mathrm{DAC-O}}}_{t-k}[M^\star_{t-k} \vert w_{1:t}] }  \stackrel{(d)}{\leq } \norm{\hat{s}_t - s_t} + R_u\sqrt{h}\norm{\widehat{G}_t - G_t}_2 \nonumber \\
& \stackrel{(e)}{\leq} 2R_u\sqrt{h}\norm{G_{t} - \widehat{G}_{t}}_2 + \frac{\kappa_a\kappa_bR_u(1-\gamma)^{h}}{\gamma}. \nonumber 
\end{align}

Here, $(d)$ follows from triangle inequality and $(e)$ follows from substituting for $\hat{s}_t$. Putting the terms together, we get that, under the same event, with probability $1-\delta/3$
\begin{align}
&  \sum_{t=1}^T f_t(M^\star_{t:t-h}\vert \widehat{G}_t , \hat{s}_{1:t}) - \sum_{t=1}^T f_t(M^\star_{t:t-h}\vert G_t , s_{1:t}) \nonumber \\
& \leq  2LRR_u\sqrt{h}\norm{G_{t} - \widehat{G}_{t}}_2 + \frac{\kappa_a\kappa_bLRR_u(1-\gamma)^{h}T}{\gamma}. \nonumber 
\end{align}

Next, we bound the final policy approximation error. Let $M_\star$ be the optimizing disturbance gain for $\sum_{t=1}^T c_t(y^{\pi}_t, u^{\pi}_t)$ and let $\pi_\star$ denote the policy corresponding to $M_\star$. Then,
\begin{align}
& \sum_{t=1}^T f_t(M^\star_{t:t-h}\vert G_t , s_{1:t}) - \min_{\pi \in \Pi_\mathcal{M}}\sum_{t=1}^T c_t(y^{\pi}_t, u^{\pi}_t) = \sum_{t=1}^T f_t(M^\star_{t:t-h}\vert G_t , s_{1:t}) - \sum_{t=1}^T c_t(y^{\pi_\star}_t, u^{\pi_\star}_t) \nonumber \\
& = \sum_{t=1}^T c_t(\tilde{\tilde{y}}^{\pi_{\mathrm{DAC-O}}}_t[M^\star_{t:t-h} \vert G_t , s_{1:t}],\tilde{u}^{\pi_{\mathrm{DAC-O}}}_t[M^\star_t\vert w_{1:t}]) - \sum_{t=1}^T c_t(y^{\pi_\star}_t, u^{\pi_\star}_t) \nonumber \\
& \stackrel{(f)}{\leq} LR\sum_{t=1}^T \left( \norm{\tilde{\tilde{y}}^{\pi_{\mathrm{DAC-O}}}_t[M^\star_{t:t-h} \vert G_t , s_{1:t}] -y^{\pi_\star}_t} + \norm{\tilde{u}^{\pi_{\mathrm{DAC-O}}}_t[M^\star_t \vert w_{1:t}] - u^{\pi_\star}_t} \right) \nonumber \\
& \stackrel{(g)}{\leq} LR\sum_{t=1}^T \left( \norm{\tilde{\tilde{y}}^{\pi_{\mathrm{DAC-O}}}_t[M^\star_{t:t-h}\vert G_t , s_{1:t}] -y^{\pi_\star}_t} \right) \stackrel{(h)}{=} LR\sum_{t=1}^T \left( \norm{\sum_{k=h+1}^{t-1} G^{[k]}_t u^{\pi_\star}_{t-k}} \right) \nonumber \\
& \stackrel{(i)}{\leq} \frac{LRR_u\kappa_a\kappa_b(1-\gamma)^hT}{\gamma}. \nonumber  
\end{align}

Here $(f)$ follows from Assumption \ref{ass:lipschitz}, $(g)$ follows from the fact that $\tilde{u}^{\pi_{\mathrm{DAC-O}}}_t[M^\star_t \vert w_{1:t}] = u^{\pi_\star}_t$ and $(h)$ follows from the fact that $\tilde{\tilde{y}}^{\pi_{\mathrm{DAC-O}}}_t[M^\star_{t:t-h} \vert G_t , s_{1:t}]$ is just the truncation of $x^{\pi_\star}_t$ and $(i)$ follows from Assumption \ref{ass:sys}. Putting together the bound on all the terms we get that with a probability $1-\delta/3$

\begin{align}
& R_T \leq \mathcal{O}\left(\sqrt{T(1+P_T)}\right)  \nonumber \\
& + 3LR\left(1+\frac{\kappa_a\kappa_b}{\gamma}\right)\sqrt{m+\log(1/\delta)}T\sigma + 2LRR_u\sqrt{h}\sum_{t=1}^T\norm{G_{t} - \widehat{G}_{t}}_2 \nonumber \\
& + \frac{2LRR_u\kappa_a\kappa_b(1-\gamma)^hT}{\gamma}. \nonumber
\end{align}

Next, we bound the term $\sum_{t=1}^T\norm{G_{t} - \widehat{G}_{t}}_2$. We denote the system parameters in a period where it does not change by $G$, without the time subscript. First we introduce two lemmas which we then use to derive the final bound.
\begin{lemma}
For any period $k$ in the CPD algorithm where the system parameters are a constant, i.e., $G_t = G$, suppose $t_p \geq chm\log(2hm)^2\log(2t_pm)^2$ for some constant $c > 0$, $T \geq 3$, then with probability $1-N^{-\log(N)}-\delta$
\begin{align}
    & \norm{\widehat{G}^{\mathrm{cd}}_k-G}_2 \leq \frac{\beta(\delta, \lambda, \sigma, h ,N)}{\sigma\sqrt{N}}, ~~ \beta(\delta, \lambda, \sigma, h ,N) = 2\sqrt{h}\zeta_\Delta\left( \sqrt{ n\log\left(2\right) + 2\log\left(\frac{2h}{\delta}\right)} + \frac{\lambda\kappa_a\kappa_b}{\gamma\zeta_\Delta\sigma\sqrt{hN}}\right) \nonumber \\
& \zeta_\Delta = \left(R_s + \frac{\kappa_a\kappa_b\kappa_m\kappa_wh}{\gamma} + \frac{\kappa_a\kappa_bR_u}{\gamma}\right). \nonumber 
\end{align}
\label{lem:esterror-k}
\end{lemma}
{\em proof}: The proof follows the proof style of \cite[Lemma D.4]{simchowitz2020improper} but requires many additional steps. We recall
\[ \widehat{G}^{\mathrm{cd}}_k = \argmin_{\widehat{G}} \sum_{p = t^k_s+h}^{t^k_e} \ell_p\left(\widehat{G}\right) + \lambda \norm{\widehat{G}}, ~~ \ell_p\left(\widehat{G}\right) = \norm{y^{\pi_{\mathrm{DAC-O}}}_p - \sum_{l=1}^h \widehat{G}^{[l]}\delta u^{\pi_{\mathrm{DAC-O}}}_{p-l}}^2.\]

Let $\delta_p = y^{\pi_{\mathrm{DAC-O}}}_p - \sum_{l=1}^h G^{[l]}\delta u^{\pi_{\mathrm{DAC-O}}}_{p-l}$. Then, it follows that 
\[ \ell_p\left(\widehat{G}\right) = \norm{\delta_p - \sum_{l=1}^h (\widehat{G}^{[l]} - G^{[l]})\delta u^{\pi_{\mathrm{DAC-O}}}_{p-l}}^2 .\]

Let $\delta \mathbf{u}^\top_p = [\delta {u^{\pi_{\mathrm{DAC-O}}}_{p-1}}^\top, \delta {u^{\pi_{\mathrm{DAC-O}}}_{p-2}}^\top, \dots, \delta {u^{\pi_{\mathrm{DAC-O}}}_{p-h}}^\top]$. Then,
\[ \ell_p\left(\widehat{G}\right) = \norm{\delta_p - (\widehat{G} - G)\delta \mathbf{u}_p}^2 .\]

Let $\mathbf{\Delta}^\top = [\delta_{t^k_s+h}, \dots, \delta_{t^k_e}], ~~ \mathbf{U}^\top = [\delta \mathbf{u}^{\pi_{\mathrm{DAC-O}}}_{t^k_s+h}, \dots, \delta \mathbf{u}^{\pi_{\mathrm{DAC-O}}}_{t^k_e}].$ Then, the solution to the least squares satisfies
\[ (\widehat{G}^{\mathrm{cd}}_k)^\top - G^\top = \left(\mathbf{U}^\top\mathbf{U}+\lambda I\right)^{-1}\left(\mathbf{U}^\top\mathbf{\Delta} -\lambda G^\top\right).\]

Then, by Cauchy-Schwarz
\[ \norm{\widehat{G}^{\mathrm{cd}}_k - G}_2 \leq  \norm{\left(\mathbf{U}^\top\mathbf{U}+\lambda I\right)^{-1}}_2\left(\norm{\mathbf{U}^\top\mathbf{\Delta}}_2 + \lambda \norm{G}_2\right).\]

Now, by expanding $\mathbf{U}^\top\mathbf{\Delta}$ it follows that the rows of $\mathbf{U}^\top\mathbf{\Delta}$ are given by
\[ \sum_{p = t^k_s+h}^{t^k_e} \delta u^{\pi_{\mathrm{DAC-O}}}_{p-i}\delta^\top_p, ~~ \forall i \in [1,h]. \]

Let $\mathcal{S}^{m} := \{v \in \mathbb{R}^m : \norm{v} = 1\}$. Then, following steps similar to \cite[Lemma D.4]{simchowitz2020improper}, which follows from standard matrix norm inequalities, we get
\[ \norm{\mathbf{U}^\top\mathbf{\Delta}}_2 \leq \sqrt{h} \max_{i \in [1,h]}\norm{\sum_{p = t^k_s+h}^{t^k_e} \delta u^{\pi_{\mathrm{DAC-O}}}_{p-i}\delta^\top_p}_2 \leq \sqrt{h} \max_{i \in [1,h]} \max_{v \in \mathcal{S}^{m-} } \norm{v^\top\sum_{p = t^k_s+h}^{t^k_e} \delta u^{\pi_{\mathrm{DAC-O}}}_{p-i}\delta^\top_p}_2. \]

By definition
\[ \delta_p = s_p + \sum_{i = 1}^{p-1}\sum_{j=1}^h G^{[i]}_pM^{[j]}_{p-i}w_{p-i-j} + \sum_{i = h+1}^{p-1}G^{[i]}_p\delta u^{\pi_{\mathrm{DAC-O}}}_{p-i}.\]

Consider the filtration $\mathcal{F}_t$ generated by the sequence of random inputs $\delta u^{\pi_{\mathrm{DAC-O}}}_t$. The terms in $\delta_p$ that are dependent on the sequence of random inputs are the third and the second term. The third term is clearly $\mathcal{F}_{p-h-1}$ measurable. The variables in second term that are dependent on the sequence of random inputs are $M_{p-i}$s. $M_{p-i}$ is dependent only through $\widehat{G}_{p-i-1}$ given the linearity of the cost functions and the update equation for $M_t$. By Eq. \eqref{eq:ls-est}, it follows that $\widehat{G}_{p-i-1}$ is only a function of the random inputs up to $p-i-h-1$ time steps. Further, by the fact that the least-squares solution is a regularized least-squares, and the fact that $\tn{Proj}_\mathcal{G}$ is also continuous (which follows from the fact that $\mathcal{G}$ is a convex set in a Hilbert space), $\widehat{G}_{p-i-1}$ is  a continuous function. By the update equation for $M_t$s, it follows that $M_{p-i}$ is also a continuous function, and therefore $\mathcal{F}_{p-i-h-1}$ measurable. Therefore, $\delta_p$ is overall $\mathcal{F}_{p-h-1}$ measurable. 

Given that $\delta u^{\pi_{\mathrm{DAC-O}}}_{p-i}$ is $\mathcal{F}_{p-i}$ is measurable and $\delta_p$ is $\mathcal{F}_{p-h-1}$ for each $i \in [1,h]$, the self-normalized martingale inequality \cite[Theorem 16]{abbasi2011regret} can be applied to the sum $v^\top\sum_{p = t^k_s+h}^{t^k_e} \delta u^{\pi_{\mathrm{DAC-O}}}_{p-i}\delta^\top_p$ by replacing $k$ in \cite[Theorem 16]{abbasi2011regret} by $p$ and setting $\eta_{p} = v^\top \delta u^{\pi_{\mathrm{DAC-O}}}_{p-i}$, $m_{p-1} = \delta^\top_p$. 

Then, by recognizing that the length of the sum $\sum_{p = t^k_s+h}^{t^k_e} \delta u^{\pi_{\mathrm{DAC-O}}}_{p-i}\delta^\top_p$ is $N$ and therefore setting $V = \zeta^{2} I, ~ V_N = \mathbf{\Delta}^\top\mathbf{\Delta}, ~ \overline{V}_N = \mathbf{\Delta}^\top\mathbf{\Delta} + V$ in \cite[Theorem 16]{abbasi2011regret}, we have that with probability $1-\delta/(2h)$
\[ \norm{v^\top\sum_{p = t^k_s+h}^{t^k_e} \delta u^{\pi_{\mathrm{DAC-O}}}_{p-i}\delta^\top_p (\mathbf{\Delta}^\top\mathbf{\Delta}+\zeta^2 I)^{-1/2}}^2_2 \leq \sigma^2\left( \log\left(\frac{\tn{det}(\mathbf{\Delta}^\top\mathbf{\Delta}+\zeta^2 I)}{\zeta^{2n}}\right) + 2\log\left(\frac{2h}{\delta}\right) \right).\]

Let $\zeta$ be a parameter corresponding to the event $\mathcal{E}_\zeta:= \norm{\mathbf{\Delta}}_2 \leq \zeta$. Then, under this event we have that 
\[ \tn{det}(\mathbf{\Delta}^\top\mathbf{\Delta}+\zeta^2 I) \leq 2^n\zeta^{2n}, ~ \tn{i.e.}, ~\frac{\tn{det}(\mathbf{\Delta}^\top\mathbf{\Delta}+\zeta^2 I)}{\zeta^{2n}} \leq 2^n. \]

By the fact that $(\mathbf{\Delta}^\top\mathbf{\Delta}+\zeta^2 I)^{-1} \leq \frac{1}{\zeta^2}$, we have that
\[ \frac{1}{\zeta^2}\norm{v^\top\sum_{p = t^k_s+h}^{t^k_e} \delta u^{\pi_{\mathrm{DAC-O}}}_{p-i}\delta^\top_p}^2_2 \leq \norm{v^\top\sum_{p = t^k_s+h}^{t^k_e} \delta u^{\pi_{\mathrm{DAC-O}}}_{p-i}\delta^\top_p (\mathbf{\Delta}^\top\mathbf{\Delta}+\zeta^2 I)^{-1/2}}^2_2. \]

Then, combining the above three equations, under the event $\mathcal{E}_\zeta$ and with probability $1-\delta/(2h)$
\[ \frac{1}{\zeta^2}\norm{v^\top\sum_{p = t^k_s+h}^{t^k_e} \delta u^{\pi_{\mathrm{DAC-O}}}_{p-i}\delta^\top_p}^2_2 \leq \sigma^2\left( n\log\left(2\right) + 2\log\left(\frac{2h}{\delta}\right) \right). \]

Next, by \cite[Lemma C.2]{oymak2019non}, given that the conditions of \cite[Lemma C.2]{oymak2019non} are satisfied, with probability at least $1-N^{-\log(N)}$
\[ \mathbf{U}^\top\mathbf{U} \geq N\sigma^2/2, ~ \tn{i.e.}, ~ \mathbf{U}^\top\mathbf{U} + \lambda I \geq N\sigma^2/2 .\]

Therefore, under the event $\mathcal{E}_\zeta$, by union bound, with probability $1-\delta/2-N^{-\log(N)}$
\begin{align}
& \norm{\widehat{G}^{\mathrm{cd}}_k - G}_2 \leq  \norm{\left(\mathbf{U}^\top\mathbf{U}+\lambda I\right)^{-1}}_2\left(\norm{\mathbf{U}^\top\mathbf{\Delta}}_2 + \lambda \norm{G}_2\right) \leq \frac{2\sqrt{h}\zeta}{N\sigma}\left( \sqrt{ n\log\left(2\right) + 2\log\left(\frac{2h}{\delta}\right)} + \frac{\lambda}{\sqrt{h}\zeta\sigma} \norm{G}_2\right) \nonumber \\
& \leq \frac{2\sqrt{h}\zeta}{N\sigma}\left( \sqrt{ n\log\left(2\right) + 2\log\left(\frac{2h}{\delta}\right)} + \frac{\lambda\kappa_a\kappa_b}{\sqrt{h}\zeta\sigma\gamma}\right). \nonumber 
\end{align} 

By \cite[Claim D.3]{simchowitz2020improper}, for $T \geq 3$, $\delta \leq 1/T$, with probability greater than $1-\delta/2$
\[ \norm{\delta u^{\pi_{\mathrm{DAC-O}}}_{p}} \leq  3\sigma\sqrt{m+\log(1/\delta)}, ~~\forall ~~ p \in [t^k_s, t^k_e]. \]

Lets call the event where $\norm{\delta u^{\pi_{\mathrm{DAC-O}}}_{p}} \leq  3\sigma\sqrt{m+\log(1/\delta)}, ~~\forall ~~ p \in [t^k_s, t^k_e]$ as $\mathcal{E}_{u,\tn{b}}$. Therefore, under the event $\mathcal{E}_{u,\tn{b}}$
\begin{align}
& \norm{\mathbf{\Delta}}_2 \leq \sqrt{N} \max_{p \in [t^k_s+h,t^k_e]} \norm{\delta_p}  = \sqrt{N} \max_{p \in [t^k_s+h,t^k_e]} \norm{s_p + \sum_{i = 1}^{p-1}\sum_{j=1}^h G^{[i]}_pM^{[j]}_{p-i}w_{p-i-j} + \sum_{i = h+1}^{p-1}G^{[i]}_p\delta u^{\pi_{\mathrm{DAC-O}}}_{p-i}} \nonumber \\
& \leq \sqrt{N} \max_{p \in [t^k_s+h,t^k_e]} \left(\norm{s_p} + \norm{\sum_{i = 1}^{p-1}\sum_{j=1}^h G^{[i]}_pM^{[j]}_{p-i}w_{p-i-j}} + \norm{\sum_{i = h+1}^{p-1}G^{[i]}_p\delta u^{\pi_{\mathrm{DAC-O}}}_{p-i}}\right) \nonumber \\
& \leq \sqrt{N} \max_{p \in [t^k_s+h,t^k_e]} \left(\norm{s_p} + \sum_{i = 1}^{p-1}\norm{G^{[i]}_p}_2\sum_{j=1}^h \norm{M^{[j]}_{p-i}}\norm{w_{p-i-j}} + \sum_{i = h+1}^{p-1}\norm{G^{[i]}_p}_2\norm{\delta u^{\pi_{\mathrm{DAC-O}}}_{p-i}}\right) \nonumber \\
& \leq \sqrt{N} \left(R_s + \frac{\kappa_a\kappa_b\kappa_m\kappa_wh}{\gamma} + \frac{\kappa_a\kappa_bR_u}{\gamma}\right). \nonumber 
\end{align}

Then, for $\zeta = \sqrt{N}\zeta_\Delta$, the probability of event $\mathcal{E}_\zeta$ is greater than $1-\delta/2$. Then, by union bound, with probability greater than $1-\delta-N^{-\log(N)}$
\[\norm{\widehat{G}^{\mathrm{cd}}_k - G}_2 \leq \frac{2\sqrt{h}\zeta_\Delta}{\sqrt{N}\sigma}\left( \sqrt{ n\log\left(2\right) + 2\log\left(\frac{2h}{\delta}\right)} + \frac{\lambda\kappa_a\kappa_b}{\gamma\zeta_\Delta\sigma\sqrt{hN}}\right) \square \nonumber \]

\begin{lemma}
From any period $k$ (of the CPD algorithm) onwards, where the system parameters are a constant, i.e., $G_t = G$ for $t \geq t^k_s$, suppose $T \geq 3$, $N_p = t-t^k_s-2h+1$ and $N_p \geq chm\log(2hm)^2\log(2(N_p+h)m)^2$ for some constant $c > 0$. Then, with probability $1-\frac{(N-h)\log{T}}{(N-h)^{\log{N-h}}}-\delta$ 
\[\norm{\widehat{G}_t-G}_2 \leq \frac{\beta(\delta, \lambda, \sigma, h ,N_p)}{\sigma\sqrt{N_p}}, ~~ \beta(\delta, \lambda, \sigma, h , N_p) = 2\sqrt{h}\zeta_\Delta\left( \sqrt{ n\log\left(2\right) + 2\log\left(\frac{2hT}{\delta}\right)} + \frac{\lambda\kappa_a\kappa_b}{\gamma\zeta_\Delta\sigma\sqrt{hN_p}}\right), \]
for all $t$ s.t. $t \geq t^k_s+h+N$, $N_p \geq chm\log(2hm)^2\log(2(N_p+h)m)^2$, where $\zeta_\Delta = \left(R_s + \frac{\kappa_a\kappa_b\kappa_m\kappa_wh}{\gamma} + \frac{\kappa_a\kappa_bR_u}{\gamma}\right)$. Moreover, under the same event,
\[ \norm{\widehat{G}^{\mathrm{cd}}_k-G}_2 \leq \frac{\beta(\delta, \lambda, \sigma, h ,N)}{\sigma\sqrt{N}}. \]
\label{lem:regesterror}
\end{lemma}
{\em Proof}: The proof follows the same outline as Lemma \ref{lem:esterror-k}. We recall that
\[ \widehat{G}_t = \tn{Proj}_{\mathcal{G}}(\widehat{G}^\star), ~ \widehat{G}^\star = \argmin_{\widehat{G}} \sum_{p = t^k_s+h}^{t-h} \ell_p\left(\widehat{G}\right) + \lambda \norm{\widehat{G}}, ~~ \ell_p\left(\widehat{G}\right) = \norm{y^{\pi_{\mathrm{DAC-O}}}_p - \sum_{l=1}^h \widehat{G}^{[l]}\delta u^{\pi_{\mathrm{DAC-O}}}_{p-l}}^2.\]

The definition of $\delta_p$ and $\delta \mathbf{u}^\top_p$ are the same as the proof of Lemma \ref{lem:esterror-k}. Lets define $\mathbf{\Delta}^\top_t = [\delta_{t^k_s+h}, \dots, \delta_{t-h}], ~~ \mathbf{U}^\top_t = [\delta \mathbf{u}^{\pi_{\mathrm{DAC-O}}}_{t^k_s+h}, \dots, \delta \mathbf{u}^{\pi_{\mathrm{DAC-O}}}_{t-h}].$ 

Lets define a Hilbert space on the set of matrices $\mathbb{R}^{n\times mh}$ with the inner product defined as
\[ \langle G,H \rangle = \max_{v \in \mathcal{S}^{mh}}(Gv)^\top(Hv).\]

This inner product induces the $2-$ norm of a matrix as a norm on the space $\mathbb{R}^{n\times mh}$. Since $\mathcal{G}$ is convex by definition, it follows that any projection on to $\mathcal{G}$ is contractive. Therefore, we have that 
\[ \norm{\widehat{G}_t - G}_2 \leq \norm{\widehat{G}^\star - G}_2 \leq  \norm{\left(\mathbf{U}^\top_t\mathbf{U}_t+\lambda I\right)^{-1}}_2\left(\norm{\mathbf{U}^\top_t\mathbf{\Delta}_t}_2 + \lambda \norm{G}_2\right).\]

Next, applying \cite[Theorem 16]{abbasi2011regret} as in Lemma \ref{lem:esterror-k}, we have that with probability $1-\delta/(2h)$, for all $t \geq t^k_s+2h$ 
\[ \norm{v^\top\sum_{p = t^k_s+h}^{t-h} \delta u^{\pi_{\mathrm{DAC-O}}}_{p-i}\delta^\top_p (\mathbf{\Delta}^\top_t\mathbf{\Delta}_t+\zeta^2_t I)^{-1/2}}^2_2 \leq \sigma^2\left( \log\left(\frac{\tn{det}(\mathbf{\Delta}^\top_t\mathbf{\Delta}_t+\zeta^2_t I)}{\zeta^{2n}_t}\right) + 2\log\left(\frac{2hT}{\delta}\right) \right).\]

Let $\zeta_t$ be a parameter and define the event $\mathcal{E}_{\zeta_t} := \norm{\mathbf{\Delta}_t}_2 \leq \zeta_t$. Next, following the same steps as the proof of Lemma \ref{lem:esterror-k}, we have that, under the event $\mathcal{E}_{\zeta_t}$ for all $t \geq t^k_s+2h$ and with probability $1-\delta/(2h)$, for all $t \geq t^k_s+2h$
\[ \frac{1}{\zeta_t^2}\norm{v^\top\sum_{p = t^k_s+h}^{t-h} \delta u^{\pi_{\mathrm{DAC-O}}}_{p-i}\delta^\top_p}^2_2 \leq \sigma^2\left( n\log\left(2\right) + 2\log\left(\frac{2hT}{\delta}\right) \right). \]

Also, by \cite[Lemma C.2]{oymak2019non}, with probability at least $1-\frac{(N-h)\log{T}}{(N-h)^{\log{N-h}}}$
\[ \mathbf{U}^\top_t\mathbf{U}_t \geq N_p\sigma^2/2, ~ \tn{i.e.}, ~ \mathbf{U}^\top_t\mathbf{U}_t + \lambda I \geq N_p\sigma^2/2, \]
for all $t$ s.t. $t \in [N+h+t^k_s,T]$, and $N_p \geq chm\log(2hm)^2\log(2(N_p+h)m)^2$. The above conclusion follows from applying \cite[Lemma C.2]{oymak2019non} to each $t$ in $t \in [N+h+t^k_s,T]$, and then taking union bound. 

Therefore, under the event $\mathcal{E}_{\zeta_t}$, by union bound,
for all $t \geq N+h+t^k_s$ such that $N_p \geq chm\log(2hm)^2\log(2(N_p+h)m)^2$, and till the underlying system parameters do not change, with probability $1-\frac{(N-h)\log{T}}{(N-h)^{\log{N-h}}} -\delta/2$ 
\begin{align}
& \norm{\widehat{G}_t - G}_2 \leq  \norm{\left(\mathbf{U}^\top_t\mathbf{U}_t+\lambda I\right)^{-1}}_2\left(\norm{\mathbf{U}^\top_t\mathbf{\Delta}_t}_2 + \lambda \norm{G}_2\right) \leq \frac{2\sqrt{h}\zeta_t}{N_p\sigma}\left( \sqrt{ n\log\left(2\right) + 2\log\left(\frac{2hT}{\delta}\right)} + \frac{\lambda}{\sqrt{h}\zeta_t\sigma} \norm{G}_2\right) \nonumber \\
& \leq \frac{2\sqrt{h}\zeta_t}{N_p\sigma}\left( \sqrt{ n\log\left(2\right) + 2\log\left(\frac{2hT}{\delta}\right)} + \frac{\lambda\kappa_a\kappa_b}{\sqrt{h}\zeta_t\sigma\gamma}\right). \nonumber 
\end{align} 

Similar to Lemma \ref{lem:esterror-k}, by \cite[Claim D.3]{simchowitz2020improper}, for $T \geq 3$, $\delta \leq 1/T$, with probability greater than $1-\delta/2$
\[ \norm{\delta u^{\pi_{\mathrm{DAC-O}}}_{p}} \leq  3\sigma\sqrt{m+\log(1/\delta)}, ~~\forall ~~ p \in [1, T]. \]

Lets call the event where $\norm{\delta u^{\pi_{\mathrm{DAC-O}}}_{p}} \leq  3\sigma\sqrt{m+\log(1/\delta)} , ~~\forall ~~ p \in [1, T]$ as $\mathcal{E}_{u,\tn{b}}$. Then, similar to Lemma \ref{lem:esterror-k}, under the event $\mathcal{E}_{u,\tn{b}}$
\[ \norm{\mathbf{\Delta}_t}_2 \leq \sqrt{N_p} \left(R_s + \frac{\kappa_a\kappa_b\kappa_m\kappa_wh}{\gamma} + \frac{\kappa_a\kappa_bR_u}{\gamma}\right), ~ \forall ~ t \geq t^k_s+2h. \]

Then, for $\zeta_t = \sqrt{N_p}\zeta_\Delta$, event $\mathcal{E}_{\zeta_t}$ holds for all $t\in [h+t^k_s,T]$ with probability $1-\delta/2$. Therefore, by union bound, with probability $1-\frac{(N-h)\log{T}}{(N-h)^{\log{N-h}}} -\delta/2$,
for all $t \geq N+h+t^k_s$ such that $N_p \geq chm\log(2hm)^2\log(2(N_p+h)m)^2$, and till the underlying system parameters do not change
\[\norm{\widehat{G}_t - G}_2 \leq \frac{2\sqrt{h}\zeta_\Delta}{\sqrt{N_p}\sigma}\left( \sqrt{ n\log\left(2\right) + 2\log\left(\frac{2hT}{\delta}\right)} + \frac{\lambda\kappa_a\kappa_b}{\gamma\zeta_\Delta\sigma\sqrt{hN_p}}\right).\]

The final step follows from the fact that $\widehat{G}^{cd}_k$ is just $\widehat{G}_t$ before the projection for $t = t^k_s+N+2h$ $\square$

\begin{lemma}
From any period $k$ (of the CPD algorithm) onwards, suppose the system changes only up to $t_c$ after $t^k_s$ and $G_t = G$ for $t \geq t_c+h$. Suppose $T \geq 3, N_p = t-t^k_s-2h+1$, and $N_p \geq chm\log(2hm)^2\log(2(N_p+h)m)^2$ for some constant $c > 0$. Then, with probability $1-\frac{(N-h)\log{T}}{(N-h)^{\log{N-h}}}-\delta$ 
\begin{align}
& \norm{\widehat{G}_t-G}_2 \leq \frac{\beta(\delta, \lambda, \sigma, h ,N_p)}{\sigma\sqrt{N_p}} + \frac{36\kappa_a \kappa_b (t_c-t^k_s-h)(m+\log(1/\delta))}{\gamma N_p}, ~~ \nonumber\\
& \beta(\delta, \lambda, \sigma, h , N_p) = 2\sqrt{h}\zeta_\Delta\left( \sqrt{ n\log\left(2\right) + 2\log\left(\frac{2hT}{\delta}\right)} + \frac{\lambda\kappa_a\kappa_b}{\gamma\zeta_\Delta\sigma\sqrt{hN_p}}\right). \nonumber 
\end{align}
for all $t$ such that $t \geq t^k_s+N+h, t \geq t_c+h$, $N_p \geq chm\log(2hm)^2\log(2(N_p+h)m)^2$, where $\zeta_\Delta = \left(R_s + \frac{\kappa_a\kappa_b\kappa_m\kappa_wh}{\gamma} + \frac{\kappa_a\kappa_bR_u}{\gamma}\right)$. 
\label{lem:regesterror-limvar}
\end{lemma}
{\em proof}: The proof follows the same outline as Lemma \ref{lem:esterror-k}. We recall that
\[ \widehat{G}_t = \tn{Proj}_{\mathcal{G}}(\widehat{G}^\star), ~ \widehat{G}^\star = \argmin_{\widehat{G}} \sum_{p = t^k_s+h}^{t-h} \ell_p\left(\widehat{G}\right) + \lambda \norm{\widehat{G}}, ~~ \ell_p\left(\widehat{G}\right) = \norm{y^{\pi_{\mathrm{DAC-O}}}_p - \sum_{l=1}^h \widehat{G}^{[l]}\delta u^{\pi_{\mathrm{DAC-O}}}_{p-l}}^2.\]

The definition of $\delta \mathbf{u}^\top_p, \mathbf{U}^\top_t, \delta_p$ and $\mathbf{\Delta}^\top_t$ are the same as the proof of Lemma \ref{lem:regesterror}. As in the proof of Lemma \ref{lem:regesterror}, the contractive property of projection implies that
\[ \norm{\widehat{G}_t - G}_2 \leq \norm{\widehat{G}^\star - G}_2 \leq  \norm{\left(\mathbf{U}^\top_t\mathbf{U}_t+\lambda I\right)^{-1}}_2\left(\norm{\mathbf{U}^\top_t\mathbf{\Delta}_t}_2 + \lambda \norm{G}_2\right).\]

Now, we can rewrite $\delta_p$ as 
\[ \delta_p = y^{\pi_{\mathrm{DAC-O}}}_p + \sum_{l=1}^h (G^{[l]}_p - G^{[l]}) \delta u^{\pi_{\mathrm{DAC-O}}}_{p-l} - \sum_{l=1}^h G^{[l]}_p \delta u^{\pi_{\mathrm{DAC-O}}}_{p-l}. \]

Let $\tilde{\delta}_p = y^{\pi_{\mathrm{DAC-O}}}_p - \sum_{l=1}^h G^{[l]}_p\delta u^{\pi_{\mathrm{DAC-O}}}_{p-l}, ~~ \widetilde{\mathbf{\Delta}}^\top_t = [\tilde{\delta}_{t^k_s+h}, \dots, \tilde{\delta}_{t-h}]$ and
\[ \mathbf{\Gamma}^\top_t = [(G_{t^k_s+h} - G)\delta \mathbf{u}_{t^k_s+h}, \dots, (G_{t_c-1} - G)\delta \mathbf{u}_{t_c-1}, \underbrace{\mathbf{0}, \dots, \mathbf{0}}_{t-h-t_c+1 ~~ \tn{times}}], \]
where $\mathbf{0}$ is a vector with zeros of appropriate dimension. Then, it follows that 
\begin{align} 
& \norm{\widehat{G}_t - G}_2 \leq  \norm{\left(\mathbf{U}^\top_t\mathbf{U}_t+\lambda I\right)^{-1}}_2\left(\norm{\mathbf{U}^\top_t\widetilde{\mathbf{\Delta}}_t}_2 + \norm{\mathbf{U}^\top_t\mathbf{\Gamma}_t}_2 + \lambda \norm{G}_2\right) \nonumber \\
& = \underbrace{\norm{\left(\mathbf{U}^\top_t\mathbf{U}_t+\lambda I\right)^{-1}}_2\left(\norm{\mathbf{U}^\top_t\widetilde{\mathbf{\Delta}}_t}_2 + \lambda \norm{G}_2\right)}_{\tn{Term I}} + \underbrace{\norm{\left(\mathbf{U}^\top_t\mathbf{U}_t+\lambda I\right)^{-1}}_2 \norm{\mathbf{U}^\top_t\mathbf{\Gamma}_t}_2}_{\tn{Term II}}. \nonumber 
\end{align}

Now, Term I can be bound as in Lemma \ref{lem:regesterror}. The term $\widetilde{\mathbf{\Delta}}_t$ satisfies all the properties that $\mathbf{\Delta}_t$ in Lemma \ref{lem:regesterror} satisfies in order to apply the self-normalized martingale inequality \cite[Theorem 16]{abbasi2011improved}. Therefore, for all $t$ s.t. $t \geq t^k_s+N+h, t \geq t_c+h$, $N_p \geq chm\log(2hm)^2\log(2(N_p+h)m)^2$, under an event with probability $1-\frac{(N-h)\log{T}}{(N-h)^{\log{N-h}}} -\delta$,  
\[\tn{Term I} \leq \frac{2\sqrt{h}\zeta_\Delta}{\sqrt{N_p}\sigma}\left( \sqrt{ n\log\left(2\right) + 2\log\left(\frac{2hT}{\delta}\right)} + \frac{\lambda\kappa_a\kappa_b}{\gamma\zeta_\Delta\sigma\sqrt{hN_p}}\right).\] 
Under this event, as in Lemma \ref{lem:regesterror}, it holds that $\norm{\delta u^{\pi_{\mathrm{DAC-O}}}_{p}} \leq  3\sigma\sqrt{m+\log(1/\delta)}, ~~\forall ~~ p \in [1, T]$. Therefore,
\begin{align}
& \mathbf{U}^\top_t\mathbf{\Gamma}_t = \sum_{p = t^k_s+h}^{t_c-1} \delta \mathbf{u}_p\delta \mathbf{u}^\top_p (G_p - G)^\top \Rightarrow \norm{\mathbf{U}^\top_t\mathbf{\Gamma}_t}_2 \leq  \sum_{p = t^k_s+h}^{t_c-1} \norm{\delta \mathbf{u}_p\delta \mathbf{u}^\top_p}_2  \norm{G_p - G}_2  \nonumber \\
& \leq \frac{2\kappa_a \kappa_b}{\gamma} \sum_{p = t^k_s+h}^{t_c-1} \norm{\delta \mathbf{u}_p\delta \mathbf{u}^\top_p}_2 =  \frac{2\kappa_a \kappa_b}{\gamma} \sum_{p = t^k_s+h}^{t_c-1} \norm{\delta \mathbf{u}_p}^2_2 \leq \frac{18\kappa_a \kappa_b (t_c-t^k_s-h)\sigma^2(m+\log(1/\delta))}{\gamma}. \nonumber 
\end{align}

Under the same event, as in Lemma \ref{lem:regesterror}, it also holds that $\mathbf{U}^\top_t\mathbf{U}_t + \lambda I \geq N_p\sigma^2/2$. Therefore, under the same event,
\[ \tn{Term II} \leq \frac{36\kappa_a \kappa_b (t_c-t^k_s-h)(m+\log(1/\delta))}{\gamma N_p}. \]

Therefore, combining Term I and Term II, with probability $1-\frac{(N-h)\log{T}}{(N-h)^{\log{N-h}}} -\delta$,  
\[ \norm{\widehat{G}_t - G}_2 \leq \frac{2\sqrt{h}\zeta_\Delta}{\sqrt{N_p}\sigma}\left( \sqrt{ n\log\left(2\right) + 2\log\left(\frac{2hT}{\delta}\right)} + \frac{\lambda\kappa_a\kappa_b}{\gamma\zeta_\Delta\sigma\sqrt{hN_p}}\right) + \frac{36\kappa_a \kappa_b (t_c-t^k_s-h)(m+\log(1/\delta))}{\gamma N_p}, \]
for all $t$ s.t. $t \geq t^k_s+N+h, t \geq t_c+h$, $N_p \geq chm\log(2hm)^2\log(2(N_p+h)m)^2$ $\square$

We now present a final lemma which bounds the error of the least square estimate from the average of the system parameters over a period. We denote the average of the system parameters over a period by $\widetilde{G}_t$.
\begin{lemma}
From any period $k$ (of the CPD algorithm) onwards, suppose $N_p = t-t^k_s-2h+1, N_p \geq chm\log(2hm)^2\log(2(N_p+h)m)^2$ for some constant $c > 0$, 
$T \geq 3$, then, for any $\widetilde{G}_t$ with probability $1-\frac{(N-h)\log{T}}{(N-h)^{\log{N-h}}}-\delta$ 
\begin{align}
& \norm{\widehat{G}_t-\widetilde{G}_t}_2 \leq \frac{\beta(\delta, \lambda, \sigma, h ,N_p)}{\sigma\sqrt{N_p}} + \frac{18(m+\log(1/\delta))\sum_{p=t^k_s+h}^{t-h}\norm{G_p-\widetilde{G}_t}_2}{N_p}, ~~ \nonumber\\
& \beta(\delta, \lambda, \sigma, h , N_p) = 2\sqrt{h}\zeta_\Delta\left( \sqrt{ n\log\left(2\right) + 2\log\left(\frac{2hT}{\delta}\right)} + \frac{\lambda\kappa_a\kappa_b}{\gamma\zeta_\Delta\sigma\sqrt{hN_p}}\right). \nonumber 
\end{align}
for all $t$ such that $t \geq t^k_s+N+h$, and $N_p \geq chm\log(2hm)^2\log(2(N_p+h)m)^2$, where $\zeta_\Delta = \left(R_s + \frac{\kappa_a\kappa_b\kappa_m\kappa_wh}{\gamma} + \frac{\kappa_a\kappa_bR_u}{\gamma}\right)$.
\label{lem:regesterror-var}
\end{lemma}
{\em proof}: The proof follows the exact same steps as Lemma \ref{lem:regesterror-limvar}. The variables $\delta \mathbf{u}^\top_p, \mathbf{U}^\top_t, \delta_p, \mathbf{\Delta}^\top_t, \tilde{\delta}_p$ and $\widetilde{\mathbf{\Delta}}^\top_t$ are all the same with $G$ replaced by $\widetilde{G}_t$. $\mathbf{\Gamma}^\top_t$ is defined by
\[ \mathbf{\Gamma}^\top_t = [(G_{t^k_s+h} - \widetilde{G}_t)\delta \mathbf{u}_{t^k_s+h}, \dots, (G_{t-h} - \widetilde{G}_t)\delta \mathbf{u}_{t-h}]. \]

Then, following the steps as in Lemma \ref{lem:regesterror-limvar}, we have that
\begin{align} 
& \norm{\widehat{G}_t - \widetilde{G}_t}_2 \leq  \norm{\left(\mathbf{U}^\top_t\mathbf{U}_t+\lambda I\right)^{-1}}_2\left(\norm{\mathbf{U}^\top_t\widetilde{\mathbf{\Delta}}_t}_2 + \norm{\mathbf{U}^\top_t\mathbf{\Gamma}_t}_2 + \lambda \norm{\widetilde{G}_t}_2\right) \nonumber \\
& = \underbrace{\norm{\left(\mathbf{U}^\top_t\mathbf{U}_t+\lambda I\right)^{-1}}_2\left(\norm{\mathbf{U}^\top_t\widetilde{\mathbf{\Delta}}_t}_2 + \lambda \norm{\widetilde{G}_t}_2\right)}_{\tn{Term I}} + \underbrace{\norm{\left(\mathbf{U}^\top_t\mathbf{U}_t+\lambda I\right)^{-1}}_2 \norm{\mathbf{U}^\top_t\mathbf{\Gamma}_t}_2}_{\tn{Term II}}. \nonumber 
\end{align}

The final result follows from bounding Term I and Term II as in Lemma \ref{lem:regesterror-limvar}. The bound on Term I is the same as in  Lemma \ref{lem:regesterror-limvar}. Under the same event where the Term I bound holds, as in Lemma \ref{lem:regesterror-limvar}, we have that,
\begin{align}
& \mathbf{U}^\top_t\mathbf{\Gamma}_t = \sum_{p = t^k_s+h}^{t-h} \delta \mathbf{u}_p\delta \mathbf{u}^\top_p (G_p - \widetilde{G}_t)^\top \Rightarrow \norm{\mathbf{U}^\top_t\mathbf{\Gamma}_t}_2 \leq  \sum_{p = t^k_s+h}^{t-h} \norm{\delta \mathbf{u}_p\delta \mathbf{u}^\top_p}_2  \norm{G_p - \widetilde{G}_t}_2  \nonumber \\
& \leq \sum_{p = t^k_s+h}^{t-h} \norm{\delta \mathbf{u}_p\delta \mathbf{u}^\top_p}_2\norm{G_p - \widetilde{G}_t}_2  =  \sum_{p = t^k_s+h}^{t-h} \norm{\delta \mathbf{u}_p}^2_2\norm{G_p - \widetilde{G}_t}_2 \leq 9\sigma^2(m+\log(1/\delta))\sum_{p = t^k_s+h}^{t-h} \norm{G_p - \widetilde{G}_t}_2 . \nonumber 
\end{align}

Therefore, under the same event,
\[  \tn{Term II} \leq \frac{18(m+\log(1/\delta))\sum_{p = t^k_s+h}^{t-h} \norm{G_p - \widetilde{G}_t}_2}{N_p}. \]

Combining bounds for Term I and Term II we get the final result. This completes the proof of this lemma $\square$

Let, only $\widetilde{\gamma}$ of the changes occur within any consecutive $2t_p$ time steps. We call the time intervals where a change occurs only after $2t_p$ time steps as stationary intervals. Let us compactly denote $\beta = \beta(\delta/T, \lambda, \sigma, h ,N)$.

Now, for a given $d$ and $T$ sufficiently large, there is trivially a constant $c > 0$ such that $N \geq chm\log(2hm)^2\log(2(T+h)m)^2$. Next, we note that every interval in the CPD algorithm falls into the scenario of one of either Lemma \ref{lem:esterror-k} or Lemma \ref{lem:regesterror} or Lemma \ref{lem:regesterror-limvar} or Lemma \ref{lem:regesterror-var}. Therefore, by union bound, i.e., with probability greater than $1-\frac{T\log{T}}{(N-h)^{\log{N-h}}}-\delta$, the conclusions drawn in Lemmas \ref{lem:esterror-k}, \ref{lem:regesterror}, \ref{lem:regesterror-limvar}, \ref{lem:regesterror-var}, whichever applies, hold for any interval, with $\delta$ in these lemmas replaced by $\delta/T$. In the following, let $\mathcal{E}$ denote the event under which the previous statement holds. Then, it is clear that, $\mathcal{E}$ occurs with probability $1-\frac{T\log{T}}{(N-h)^{\log{N-h}}}-\delta$.

Next, we bound the term $\sum_{t=1}^T\norm{G_{t} - \widehat{G}_{t}}_2$. We derive the bound by the following steps: first we derive the bound for the case where all changes are large, then we derive the bound for the case where all the changes are small, and we derive the final bound by combining these two cases.

Consider a stationary interval $s$. Let $G$ denote the underlying system parameter in the current stationary interval. Let $G_{-}$ denote the underlying system parameter in the previous stationary interval. Let 
\[ \Delta = G - G_{-}, ~~ \Delta_n = \norm{\Delta}_2.\] 

\textbf{Scenario 1: $\Delta_n > \frac{4\beta}{\sigma \sqrt{N}}$, and a change point detection in the stationary interval $s$:}

Let the end of the previous stationary interval be $t^s_{e-}$ and the end of $s$ be $t^s_{e}$. Let the beginning of the interval $s$ be $t^s_{b}$. Let the change point detection occur at $t_d$ within $s$. There are two possible {\it cases} within this scenario: (i) No change point detection occurs in between the stationary interval $s$ and the previous stationary interval, (ii) multiple change point detection occurs in between the stationary interval $s$ and the previous stationary interval. 

{\it Case (i)}: This is the simplest of the cases. Let $t_d$ be the time of change point detection within $s$. Note that this has to be the beginning of some period of the CPD algorithm. Let $t_{d-}$ be the time of change point detection before $t_d$. In this case, a change point detection will not occur from $t_d$ till $t^s_{e}$ under event $\mathcal{E}$. This is because under the event $\mathcal{E}$, for any period $k$ and $l$ of the CPD algorithm that lies within the duration from $t_d$ till $t^s_{e}$
\[ \norm{\widehat{G}^{\mathrm{cd}}_{k} - \widehat{G}^{\mathrm{cd}}_l}_2 = \norm{\widehat{G}^{\mathrm{cd}}_{k} - G + G - \widehat{G}^{\mathrm{cd}}_l}_2 \leq \norm{\widehat{G}^{\mathrm{cd}}_{k} - G}_2 + \norm{G- \widehat{G}^{\mathrm{cd}}_l}_2 \leq \frac{2\beta}{\sigma N}, ~ \forall ~ l < k,\]
and thence do not satisfy the change point detection criteria. Therefore, under the event $\mathcal{E}$, from $t_d$ till $t^s_{e}$, the estimation algorithm will be run without any interruption. Let $\tilde{h} = chm\log(2hm)^2\log(2(T+h)m)^2, t_h = t_d+N+2h$. By definition, we note that $N \geq \tilde{h} \geq chm\log(2hm)^2\log(2(N+h)m)^2$. Therefore, 
\begin{align}
 & \sum_{t = t^s_{e-}+1}^{t^s_{e}}\norm{\widehat{G}_{t} - G_t}_2 \leq  \sum_{t = t^s_{e-}+1}^{t_d}\norm{\widehat{G}_{t} - G_t}_2 + \sum_{t = t_d+1}^{t^s_{e}}\norm{\widehat{G}_{t} - G_t}_2 \nonumber \\
 & = \sum_{t = t^s_{e-}+1}^{t_d}\norm{\widehat{G}_{t} - G_t}_2 + \sum_{t = t_d+1}^{t_h-1} \norm{\widehat{G}_{t} - G_t}_2 + \sum_{t = t_h}^{t^s_{e}} \norm{\widehat{G}_{t} - G_t}_2.  \nonumber
\end{align}

Let $k_1$ denote the index of the first period or block of length $t_p$ of the CPD algorithm that is entirely within the current stationary interval. We can make the following observation. 

{\it Sub-case}: If $t_{d-}$ is before the last block of length $t_p$ of the CPD algorithm within the previous stationary interval, then $t_d$ will necessarily occur at the end of $k_1$ period under event $\mathcal{E}$. This is because, under the event $\mathcal{E}$, for any block or period $k$ of the CPD algorithm that lies within the previous stationary interval, 
\begin{align} 
& \norm{\widehat{G}^{\mathrm{cd}}_{k_1} - \widehat{G}^{\mathrm{cd}}_k}_2 = \norm{\widehat{G}^{\mathrm{cd}}_{k_1} -G + \Delta + G_{-} -   \widehat{G}^{\mathrm{cd}}_k}_2 \nonumber \\
& \geq \Delta_n - \norm{\widehat{G}^{\mathrm{cd}}_{k_1} -G}_2 - \norm{G_{-} -   \widehat{G}^{\mathrm{cd}}_{k}}_2 > \frac{2\beta}{\sigma \sqrt{N}}.  \nonumber
\end{align}

It follows that, if $t_{d-}$ is before the last block of length $t_p$ of the CPD algorithm within the previous stationary interval, by the previous equation and the change point detection rule, change point detection would necessarily occur by the end of the $k_1$ period. Then, under the event $\mathcal{E}$,
\begin{align}
 & \sum_{t = t^s_{e-}+1}^{t^s_{e}}\norm{\widehat{G}_{t} - G_t}_2 \leq  \sum_{t = t^s_{e-}+1}^{t_d}\norm{\widehat{G}_{t} - G_t}_2 + \sum_{t = t_d+1}^{t_h-1} \norm{\widehat{G}_{t} - G_t}_2 + \sum_{t = t_h}^{t^s_{e}} \norm{\widehat{G}_{t} - G_t}_2  \nonumber \\
 & \stackrel{(a)}{\leq} \frac{2\kappa_a\kappa_b(t_d-t^s_{e-})}{\gamma} + \frac{2\kappa_a\kappa_b(N+2h)}{\gamma} + \sum_{t = t_h}^{t^s_{e}}\frac{\beta(\delta/T, \lambda, \sigma, h ,N_p)}{\sigma\sqrt{N_p}} \nonumber \\
 & \leq \frac{2\kappa_a\kappa_b(t_d-t^s_{e-})}{\gamma} + \frac{2\kappa_a\kappa_b(N+2h)}{\gamma} + \mathcal{O}\left( \frac{\beta\sqrt{t^s_{e}-t_d-2h}}{\sigma} \right). \nonumber 
\end{align}
Here, $(a)$ follows from applying Lemma \ref{lem:regesterror} to the last term with $\delta/(T)$ in place of $\delta$ and the first two terms follow by using the bound on the set within which $\widehat{G}_{t}$ and $G_t$ lie.

Let $\widetilde{\gamma}^s_c$ be the number of changes in the system between $s$ and the previous stationary interval. These $\widetilde{\gamma}^s_c$ changes, by the scenario under consideration, should occur within a gap of $2t_p$ time steps. This and the fact that $t_d$ should occur at the end of $k_1$, it follows that $t_d-t^s_{e-} \leq 2(\widetilde{\gamma}^s_c+1)t_p \leq 2(\widetilde{\gamma}^s_c+1)(N+h)$. 
Therefore, under the event $\mathcal{E}$, when $t_{d-}$ is before the last block of length $t_p$ of the CPD algorithm within the previous stationary interval, 
\[ \sum_{t = t^s_{e-}+1}^{t^s_{e}}\norm{\widehat{G}_{t} - G_t}_2 \leq \mathcal{O}\left(\frac{\kappa_a\kappa_b(\widetilde{\gamma}^s_c+2)(N+h)}{\gamma}\right) + \frac{2\kappa_a\kappa_b(N+2h)}{\gamma} + \mathcal{O}\left( \frac{\beta\sqrt{t^s_{e}-t^s_{e-}}}{\sigma} \right). \]

{\it Sub-case}: If $t_{d-}$ is not before the last block of length $t_p$ of the CPD algorithm within the previous stationary interval, then (i) $t_{d-}$ has to be at the end of the last block within the previous stationary interval and (ii) $t_d$ need not necessarily occur by the end of the first period or block of length $t_p$ of CPD algorithm within the current stationary interval. Let $t_{h-} = N+t_{d-}+2h$ and $N_{p-} = t-t_{d-}-2h+1$. Then, under the event $\mathcal{E}$,
\begin{align}
 & \sum_{t = t^s_{e-}+1}^{t^s_{e}}\norm{\widehat{G}_{t} - G_t}_2 \leq  \sum_{t = \max\{t_{d-},t^s_{e-}+1\}}^{t_d}\norm{\widehat{G}_{t} - G_t}_2 + \sum_{t = t_d+1}^{t^s_{e}}\norm{\widehat{G}_{t} - G_t}_2 \nonumber \\
 & = \sum_{t = \max\{t_{d-},t^s_{e-}+1\}}^{t_d}\norm{\widehat{G}_{t} - G_t}_2 + \sum_{t = t_d+1}^{t_h-1} \norm{\widehat{G}_{t} - G_t}_2 + \sum_{t = t_h}^{t^s_{e}} \norm{\widehat{G}_{t} - G_t}_2  \nonumber \\
 & \stackrel{(b)}{\leq} \sum_{t = \max\{t_{d-},t^s_{e-}+1\}}^{t_{h-}-1} \norm{\widehat{G}_{t} - G_t}_2 + \sum_{t = \max\{t_{h-},t^s_{e-}+1\}}^{t_d}\norm{\widehat{G}_{t} - G_t}_2 + \frac{2\kappa_a\kappa_b(N+2h)}{\gamma} + \mathcal{O}\left( \frac{\beta\sqrt{t^s_{e}-t_d-2h}}{\sigma} \right) \nonumber  \\
 & \stackrel{(c)}{\leq} \frac{4\kappa_a\kappa_b(N+2h)}{\gamma} + \mathcal{O}\left(\frac{\kappa_a\kappa_b(\widetilde{\gamma}^s_c)(N+h)}{\gamma}\right) + \sum_{t = \max\{t_{h-},t^s_{e-}+1\}}^{t_d}\norm{\widehat{G}_{t} - G_t}_2 + \mathcal{O}\left( \frac{\beta\sqrt{t^s_{e}-t_d-2h}}{\sigma} \right) \nonumber \\
 & \stackrel{(d)}{\leq} \frac{4\kappa_a\kappa_b(N+2h)}{\gamma} + \mathcal{O}\left(\frac{\kappa_a\kappa_b(\widetilde{\gamma}^s_c)(N+h)}{\gamma}\right) + \mathcal{O}\left( \frac{\beta\sqrt{t^s_{e}-t_d-2h}}{\sigma} \right) \nonumber \\
 & + \sum_{t = \max\{t_{h-},t^s_{e-}+1\}}^{t_d} \left(\frac{\beta(\delta/T, \lambda, \sigma, h ,N_{p-})}{\sigma\sqrt{N_{p-}}}  + \frac{72\kappa_a \kappa_b (\widetilde{\gamma}^s_c+1)(N+h)(m+\log(T/\delta))}{\gamma N_{p-}}\right) \nonumber \\
 & \leq \frac{4\kappa_a\kappa_b(N+2h)}{\gamma} + \mathcal{O}\left(\frac{\kappa_a\kappa_b(\widetilde{\gamma}^s_c+2)(N+h)(1+(m+\log(T/\delta))\log(T))}{\gamma}\right)  \nonumber \\
 & + \mathcal{O}\left( \frac{\beta\sqrt{t^s_{e}-t_d-2h}}{\sigma} \right) + \sum_{t = \max\{t_{h-},t^s_{e-}+1\}}^{t_d} \left(\frac{\beta(\delta/T, \lambda, \sigma, h ,N_{p-})}{\sigma\sqrt{N_{p-}}}\right) \nonumber \\
  & \leq \frac{4\kappa_a\kappa_b(N+2h)}{\gamma} + \mathcal{O}\left(\frac{\kappa_a\kappa_b(\widetilde{\gamma}^s_c+2)(N+h)(1+(m+\log(T/\delta))\log(T))}{\gamma}\right)  \nonumber \\
 & + \mathcal{O}\left( \left(\frac{\beta}{\sigma} \right) \left(\sqrt{t^s_{e}-t_d-2h} + \sqrt{t_{d}-t^s_{e-}-2h}\right)\right).
\end{align}
Here, $(b)$ follows from applying Lemma \ref{lem:regesterror} with $\delta/(T)$ in place of $\delta$ to the last term and applying the bound on the set within which $\widehat{G}_t$ and $G_t$ lie to the second term, the first term in $(c)$ follows from applying the bound on $\widehat{G}_t$ and $G_t$ to the first term in $(b)$, the second term in $(c)$ follows from repeating the same to the first $t_c$ time steps from $t_{h-}$ as defined in Lemma \ref{lem:regesterror-limvar}, which in this case is $2\widetilde{\gamma}^s_c(N+h)$: follows from the fact that there can be at the most $\widetilde{\gamma}^s_c$ between the stationary intervals, $(d)$ follows from applying Lemma \ref{lem:regesterror-limvar} with $\delta/(T)$ in place of $\delta$. The term $\left(\sqrt{t^s_{e}-t_d-2h} + \sqrt{t_{d}-t^s_{e-}-2h}\right)$ is maximum when $t^s_{e}-t_d = t_{d}-t^s_{e-}$. Therefore, under event $\mathcal{E}$, 
\begin{align}
 & \sum_{t = t^s_{e-}+1}^{t^s_{e}}\norm{\widehat{G}_{t} - G_t}_2 
  \leq \frac{4\kappa_a\kappa_b(N+2h)}{\gamma} + \mathcal{O}\left(\frac{\kappa_a\kappa_b(\widetilde{\gamma}^s_c+2)(N+h)(1+(m+\log(T/\delta))\log(T))}{\gamma}\right)  \nonumber \\
 & + \mathcal{O}\left( \left(\frac{\beta}{\sigma} \right) \left(\sqrt{t^s_{e}-t^s_{e-}}\right)\right). \nonumber 
\end{align}

{\it Case (ii)}: The steps involved in bounding this {\it case} are same as the second sub-case of {\it case (i)}. Therefore, we get the same bound. This covers all the cases for this scenario. 

Let $R^g_1(t^s_{e-},t^s_{e})$ denote the cumulative error in this scenario. Then, $R^g_1(t^s_{e-},t^s_{e})$ is given by the bound derived for the second sub-case of case (i) of scenario 1.

\textbf{Scenario 2: $\Delta_n > \frac{4\beta}{\sigma \sqrt{N}}$, and no change point detection in the stationary interval $s$:} 

If there is no change point detection at all inside the stationary interval, then under event $\mathcal{E}$, given that $\Delta_n > \frac{4\beta}{\sigma \sqrt{N}}$, there should have been a change point detection at the end of the last period or block of length $t_p$ within the previous stationary interval. Since there is no change point detection in $s$, this scenario is equivalent to second sub-case of {\it case (i)} of scenario 1 with $t_d$ in that bound set to $t^s_{e}$. Therefore, the bound derived over there applies here as well. Let $R^g_2(t^s_{e-},t^s_{e})$ denote the cumulative error in this scenario. Then, $R^g_2(t^s_{e-},t^s_{e})$ is given by the bound derived for the second sub-case of case (i) of scenario 1.

\textbf{Scenario 3: $\Delta_n \leq \frac{4\beta}{\sigma \sqrt{N}}$, and there is no change point detection in the stationary interval $s$}:

Under event $\mathcal{E}$, by the same argument as in {\it case (i)} of scenario 1, there can only be at the most one change point detection within any stationary interval. Let $t_{d}$ be the change point before an interval $s$ (satisfying this scenario) and $t_{d+}$ be the first change point detection starting from $s$. Let there be $q$ non-stationary intervals between these two change point detections. With a slight abuse of notation, denote the total number of changes within the non-stationary intervals in this period together by $\widetilde{\gamma}^s_c$.

We denote the true underlying model in each stationary interval in this period by $G_{r}$, where $r$ is the index of the stationary interval that overlaps with the period between the two aforementioned change point detections. We remind that this changes the notation for the true underlying model we introduced in the beginning for the scenario under consideration. Then, under event $\mathcal{E}$, for any two pair of stationary intervals $r$ and $r'$ within this period, whose overlap with this period is greater than $t_p$,
\[ \norm{G_r-G_{r'}}_2 \leq \frac{4\beta}{\sigma \sqrt{N}}. \]
Otherwise, under event $\mathcal{E}$, there would have to be a change point detection in at least one of these intervals.

Let the start and end time of the $r$th stationary interval be $t^r_s$ and $t^r_e$. Let $\tilde{h} = chm\log(2hm)^2\log(2(T+h)m)^2, t_h = t_{d}+N+2h$. Then, under event $\mathcal{E}$, 
\begin{align}
 & \sum_{t = t^{1}_e+1}^{t^{q+1}_{e}}\norm{\widehat{G}_{t} - G_t}_2 = \sum_{t = t^1_e+1}^{t^2_s-1}\norm{\widehat{G}_{t} - G_t}_2 + \sum_{t = t^2_s}^{t^2_e}\norm{\widehat{G}_{t} - G_t}_2  + \sum_{t = t^q_e+1}^{t^{q+1}_{s}-1}\norm{\widehat{G}_{t} - G_t}_2 + \sum_{t = t^{q+1}_{s}}^{t^{q+1}_{e}}\norm{\widehat{G}_{t} - G_t}_2 \nonumber \\
 & \stackrel{(e)}{\leq} \mathcal{O}\left(\frac{\kappa_a\kappa_b(\widetilde{\gamma}^s_c+2)(N+h)}{\gamma}\right) + \sum_{t = t^2_s}^{t^2_e}\norm{\widehat{G}_{t} - G_t}_2 \dots + \sum_{t = t^{q+1}_{s}}^{t^{q+1}_{e}}\norm{\widehat{G}_{t} - G_t}_2 \nonumber \\
 & \leq \mathcal{O}\left(\frac{\kappa_a\kappa_b(\widetilde{\gamma}^s_c+2)(N+h)}{\gamma}\right) + \sum_{t = t^1_{e}+1}^{t_h-1}\norm{\widehat{G}_{t} - G_t}_2 + \sum_{t = t^2_s}^{t^2_e}\norm{\widehat{G}_{t} - G_t}_2  + \sum_{t = t^{q+1}_{s}}^{t^{q+1}_{e}}\norm{\widehat{G}_{t} - G_t}_2 \nonumber \\
  & \stackrel{(f)}{\leq} \mathcal{O}\left(\frac{\kappa_a\kappa_b(\widetilde{\gamma}^s_c+2)(N+h)}{\gamma}\right) + \frac{2\kappa_a\kappa_b(N+2h)}{\gamma} + \sum_{t = t^2_s}^{t^2_e}\norm{\widehat{G}_{t} - G_t}_2 + \sum_{t = t^{q+1}_s+1}^{t_{d+}}\norm{\widehat{G}_{t} - G_t}_2 + \sum_{t = t_{d+}+1}^{t^{q+1}_{e}}\norm{\widehat{G}_{t} - G_t}_2. \nonumber 
\end{align}

Here, $(e)$ follows from accumulating errors from all the non-stationary intervals in this period, $(f)$ follows from applying the bound of the set within which $\widehat{G}_t$ and $G_t$ lie to the second term in the previous line. We denote the total number of changes within the non-stationary interval before the $r$th stationary interval by $\widetilde{\gamma}^r_c$. Then, under event $\mathcal{E}$, the cumulative error within the period between the change point detections at $t_{d}$ and $t_{d+}$
\begin{align}
& \sum_{t = \max\{t^r_s,t_h\}}^{\min\{t^r_e, t_{d+}\}}\norm{\widehat{G}_{t} - G_t}_2 \nonumber \\
& \stackrel{(g)}{\leq} \sum_{t = \max\{t^r_s,t_h\}}^{\min\{t^r_e, t_{d+}\}}\left(\frac{\beta(\delta/T, \lambda, \sigma, h , N_p)}{\sigma\sqrt{N_p}} + \frac{18(m+\log(T/\delta))\sum_{p=t_d+h}^{t-h}\norm{G_p-G_t}_2}{N_p}\right) \nonumber \\
& \leq \sum_{t = \max\{t^r_s,t_h\}}^{\min\{t^r_e, t_{d+}\}}\left(\frac{\beta}{\sigma\sqrt{N_p}} + \frac{18(m+\log(T/\delta))\sum_{p=t_d+h}^{t-h}\norm{G_p-G_t}_2}{N_p}\right)  \nonumber \\
& = \sum_{t = \max\{t^r_s,t_h\}}^{\min\{t^r_e, t_{d+}\}}\left(\frac{\beta}{\sigma\sqrt{N_p}}\right) + \sum_{t = \max\{t^r_s,t_h\}}^{\min\{t^r_e, t_{d+}\}}\left(\frac{18(m+\log(T/\delta))\sum_{p=t_d+h}^{t-h}\norm{G_p-G_t}_2}{N_p}\right) \nonumber \\
& \stackrel{(h)}{\leq}\sum_{t = \max\{t^r_s,t_h\}}^{\min\{t^r_e, t_{d+}\}}\left(\frac{\beta}{\sigma\sqrt{N_p}}\right) + \sum_{t = \max\{t^r_s,t_h\}}^{\min\{t^r_e, t_{d+}\}}18(m+\log(T/\delta))\left(\frac{2\kappa_a\kappa_b(\widetilde{\gamma}^r_c)(N+h)}{\gamma N_p} + \frac{4\beta}{\sigma \sqrt{N}} \right). \nonumber 
\end{align}

Here, $(g)$ follows from applying Lemma \ref{lem:regesterror-var} with $\delta/T$ in place of $\delta$, $(h)$ follows from using the condition $\Delta_n \leq \frac{4\beta}{\sigma \sqrt{N}}$ and the bound on the set within which $G_t$s lie. Putting together the sum for all $r$s, we get that 
\begin{align}
& \sum_{t = t_{d}}^{t^{q+1}_e}\norm{\widehat{G}_{t} - G_t}_2 \leq \mathcal{O}\left(\frac{\kappa_a\kappa_b(\widetilde{\gamma}^s_c+2)(N+h)}{\gamma}\right) + \frac{2\kappa_a\kappa_b(N+2h)}{\gamma} + \mathcal{O}\left(\frac{\beta\sqrt{t^{q+1}_{e} -t^1_{e}}}{\sigma}\right) \nonumber \\
& + \frac{36\kappa_a\kappa_b(\widetilde{\gamma}^s_c)(N+h)(m+\log(T/\delta))\log(T)}{\gamma} + \frac{72(m+\log(T/\delta))\beta (t^{q+1}_{e} - t^1_{e})}{\sigma \sqrt{N}} + \sum_{t = t_{d+}+1}^{t^{q+1}_{e}}\norm{\widehat{G}_{t} - G_t}_2 \nonumber \\
& \leq \mathcal{O}\left(\frac{\kappa_a\kappa_b(\widetilde{\gamma}^s_c+2)(N+h)}{\gamma}\right) + \frac{4\kappa_a\kappa_b(N+2h)}{\gamma} + \mathcal{O}\left(\frac{\beta\sqrt{t^{q+1}_{e} -t^1_{e}}}{\sigma}\right) \nonumber \\
& + \frac{36\kappa_a\kappa_b(\widetilde{\gamma}^s_c)(N+h)(m+\log(T/\delta))\log(T)}{\gamma} + \frac{72(m+\log(T/\delta))\beta (t^{q+1}_{e} - t^1_{e})}{\sigma \sqrt{N}} + \mathcal{O}\left( \frac{\beta\sqrt{t^{q+1}_{e}-t^1_{e}}}{\sigma} \right). 
\end{align}
Let $R^g_3(t^s_{e-},t^s_{e})$ denote the cumulative error corresponding to a stationary interval that is part of a scenario such as scenario 3.

\textbf{Scenario 4: $\Delta_n \leq \frac{4\beta}{\sigma \sqrt{N}}$, and there is a change point detection in the stationary interval $s$}:

Under event $\mathcal{E}$, by the same argument as in {\it case (i)} of scenario 1, there can only be at the most one change point detection within any stationary interval. So in this scenario, there wouldn't be another change point detection. Therefore, the cumulative error over the period from $t^s_{e-}$ to $t^s_{e}$ in this scenario can be bounded similar to second sub-case of case (i) of scenario 1. Therefore, the same bound derived there applies here as well. Let $R^g_4(t^s_{e-},t^s_{e})$ denote the cumulative error corresponding to a stationary interval satisfying scenario 4.

This covers all possible scenarios that can occur in a stationary interval. Let us index the stationary intervals by $s$. Given that $\widetilde{\gamma}$ of the changes occur within $2t_p$ time steps, there will be $\Gamma_T-\widetilde{\gamma}+1$ stationary intervals. We denote the scenario of an interval $s$ by $S(s)$. Then
\begin{align}
& \sum_{t=1}^T \norm{\widehat{G}_t -G_t}_2 = \sum_{s=1}^{\Gamma_T-\widetilde{\gamma}+1} \sum_{t = t^{s-1}_{e}}^{t^s_{e}} \norm{\widehat{G}_t -G_t}_2 \nonumber \\
& = \sum_{s=1}^{\Gamma_T-\widetilde{\gamma}+1} \left( \mathbb{I}[S(s) = 1] R^g_1(t^{s-1}_{e},t^s_{e}) + \mathbb{I}[S(s) = 2]R^g_2(t^{s-1}_{e},t^s_{e}) + \mathbb{I}[S(s) = 3]R^g_3(t^{s-1}_{e},t^s_{e}) + \mathbb{I}[S(s) = 4]R^g_4(t^{s-1}_{e},t^s_{e})\right). \nonumber
\end{align}
Note that, in any control episode, as no two scenarios can co-exist in an interval $s$, under event $\mathcal{E}$, $S(s)$ can only be one of the four scenarios. Therefore, the cumulative error is bounded by
\begin{align}
& \sum_{t = 1}^{T}\norm{\widehat{G}_{t} - G_t}_2 
  \leq \frac{4\kappa_a\kappa_b(N+2h)\Gamma_T}{\gamma} + \mathcal{O}\left(\frac{\beta\sqrt{T\Gamma_T}}{\sigma}\right) \nonumber \\
& + \mathcal{O}\left(\frac{\kappa_a\kappa_b\Gamma_T(N+h)(1+(m+\log(T/\delta))\log(T))}{\gamma}\right) + \widetilde{\mathcal{O}}\left(\frac{\beta T}{\sigma \sqrt{N}}\right). \nonumber
\end{align}
Substituting the values of $\sigma$ and $N$ into this bound, and then substituting the cumulative error by this bound in the bound derived for the regret, we get the desired result for the regret with probability $1-\frac{T\log{T}}{(N-h)^{\log{N-h}}}-\delta$.

Now, given that $\Gamma_T = \widetilde{\mathcal{O}}\left(T^d\right)$, $d < 1$, and $N = \Gamma^{-0.8}_TT^{4/5}$, $N = \mathcal{O}(T^{4/5(1-d)})$. Therefore,
\beq
 \lim_{T \rightarrow \infty} T\log(T)/((N-h)^{\log{N-h}}) = 0. \nonumber 
\eeq
Therefore, for $\Gamma_T = \widetilde{\mathcal{O}}\left(T^d\right)$ and a given $\widetilde{\delta}$ such that $\delta \leq \widetilde{\delta}$, we can find a $T_1$ such that, for all $T > T_1$, $T\log(T)/((N-h)^{\log{N-h}})$ is smaller than $\widetilde{\delta}$. Also, for $\Gamma_T = \widetilde{\mathcal{O}}\left(T^d\right)$, $t_p = N+h > N \geq \mathcal{O}(T^{2/5(1-d)})$, and $hm\log(2hm)^2\log(2(T+h)m)^2 = \mathcal{O}(\log(T)^5)$. Therefore, we can find a $T_2$ and a constant $c > 0$ such that, for all $T > T_2$, $N \geq chm\log(2hm)^2\log(2(T+h)pm)^2$. Take $T_0 = \max\{T_1,T_2\}$. Then, for $T > T_0$, we can find a constant $c$ such that the condition $N \geq chm\log(2hm)^2\log(2(T+h)m)^2$ holds, and $T\log(T)/((N-h)^{\log{N-h}})$ is smaller than $\widetilde{\delta}$. Therefore, for $T > T_0$ the regret holds with probability $1-2\widetilde{\delta}$ $\blacksquare$

\section{Proof of Theorem \ref{thm:olc-zk-cd} : Setting S-1 }
\label{sec:th2-proof-1}

The proof steps are exactly same as setting S-2. The differences are minor and we only highlight the differences here without repeating all the steps. Proceeding exactly as in Setting S-2 the constant $R_s$ would be $\frac{\kappa_a\kappa_b\kappa_w}{\gamma} + \kappa_e + \frac{2R_u\kappa_a\kappa_b}{\gamma}$. The other difference occurs in the cost truncation error:
\begin{align}
& \sum_{t=1}^T c_t(y^{\pi_{\mathrm{DAC-O}}}_t, u^{\pi_{\mathrm{DAC-O}}}_t) - \sum_{t=1}^T c_t(\tilde{y}^{\pi_{\mathrm{DAC-O}}}_t[M_{t:t-h}\vert \widehat{G}_t , \hat{s}_{1:t}], u^{\pi_{\mathrm{DAC-O}}}_t) \leq LR \norm{y^{\pi_{\mathrm{DAC-O}}}_t - \tilde{y}^{\pi_{\mathrm{DAC-O}}}_t[M_{t:t-h}\vert \widehat{G}_t,\hat{s}_{1:t}] } \nonumber \\
& \leq LR(R_u + \kappa_w)\sqrt{h} \norm{G_t - \hat{G}_t} + \frac{\kappa_a\kappa_b}{\gamma}\left( \kappa_w + R_u\right). \nonumber 
\end{align}

The second difference lies in the following steps.
\begin{align}
& \norm{\tilde{\tilde{y}}^{\pi_{\mathrm{DAC-O}}}_t[M_\star \vert \widehat{G}_t , \hat{s}_{1:t}] - \tilde{\tilde{y}}^{\pi_{\mathrm{DAC-O}}}_t[M_\star \vert G_t , s_{1:t}]} \nonumber \\
& = \norm{ (\hat{s}_t - s_t) + \sum_{k=1}^h  \widehat{G}^{[k]}_t\tilde{u}^{\pi_{\mathrm{DAC-O}}}_{t-k}[M_\star \vert w_{1:t}] - \sum_{k=1}^h  G^{[k]}_t\tilde{u}^{\pi_{\mathrm{DAC-O}}}_{t-k}[M_\star \vert w_{1:t}] } \nonumber \\
& = \norm{ (\hat{s}_t - s_t) + \sum_{k=1}^h  \left(\widehat{G}^{[k]}_t - G^{[k]}_t\right)\tilde{u}^{\pi_{\mathrm{DAC-O}}}_{t-k}[M_\star \vert w_{1:t}] }  \leq  \norm{\hat{s}_t - s_t} + R_u\sqrt{h}\norm{\widehat{G}_t - G_t}_2 \nonumber \\
& \leq (R_u+\kappa_w)\sqrt{h}\norm{G_{t} - \widehat{G}_{t}}_2 + \frac{\kappa_a\kappa_bR_u(1-\gamma)^{h}}{\gamma}. \nonumber 
\end{align}

The final difference is a minor difference in Lemma \ref{lem:esterror-k}. There it is mentioned that ``$M_{p-i}$ are dependent on the random perturbation inputs only through $\widehat{G}_{p-i-1}$ given the linearity of the cost functions and the update equation for $M_t$". In Setting S-1, $M_{p-i}$s are again dependent on the random perturbation inputs only through $\widehat{G}_{p-i-1}$, but not because of linearity, but because of the definition of $\hat{s}_t$ used in setting S-1. This covers all the differences. The rest of the proof is the same.

%% file: OLC-LTV.bbl
\begin{thebibliography}{10}
\providecommand{\url}[1]{#1}
\csname url@samestyle\endcsname
\providecommand{\newblock}{\relax}
\providecommand{\bibinfo}[2]{#2}
\providecommand{\BIBentrySTDinterwordspacing}{\spaceskip=0pt\relax}
\providecommand{\BIBentryALTinterwordstretchfactor}{4}
\providecommand{\BIBentryALTinterwordspacing}{\spaceskip=\fontdimen2\font plus
\BIBentryALTinterwordstretchfactor\fontdimen3\font minus
  \fontdimen4\font\relax}
\providecommand{\BIBforeignlanguage}[2]{{%
\expandafter\ifx\csname l@#1\endcsname\relax
\typeout{** WARNING: IEEEtran.bst: No hyphenation pattern has been}%
\typeout{** loaded for the language `#1'. Using the pattern for}%
\typeout{** the default language instead.}%
\else
\language=\csname l@#1\endcsname
\fi
#2}}
\providecommand{\BIBdecl}{\relax}
\BIBdecl

\bibitem{minasyan2021online}
E.~Minasyan, P.~Gradu, M.~Simchowitz, and E.~Hazan, ``Online control of unknown
  time-varying dynamical systems,'' \emph{Advances in Neural Information
  Processing Systems}, vol.~34, 2021.

\bibitem{dean2018regret}
S.~Dean, H.~Mania, N.~Matni, B.~Recht, and S.~Tu, ``Regret bounds for robust
  adaptive control of the linear quadratic regulator,'' in \emph{NeurIPS},
  2018.

\bibitem{mania2019certainty}
H.~Mania, S.~Tu, and B.~Recht, ``Certainty equivalence is efficient for linear
  quadratic control,'' in \emph{Neural Information Processing Systems
  (NeurIPS)}, 2019.

\bibitem{agarwal2019online}
N.~Agarwal, B.~Bullins, E.~Hazan, S.~Kakade, and K.~Singh, ``Online control
  with adversarial disturbances,'' \emph{International Conference on Machine
  Learning}, pp. 111--119, 2019.

\bibitem{simchowitz2020improper}
M.~Simchowitz, K.~Singh, and E.~Hazan, ``Improper learning for non-stochastic
  control,'' \emph{Conference on Learning Theory (COLT)}, pp. 3320--3436, 2020.

\bibitem{yu2020power}
C.~Yu, G.~Shi, S.~Chung, Y.~Yue, and A.~Wierman, ``The power of predictions in
  online control,'' 2020.

\bibitem{lin2021perturbation}
Y.~Lin, Y.~Hu, G.~Shi, H.~Sun, G.~Qu, and A.~Wierman, ``Perturbation-based
  regret analysis of predictive control in linear time varying systems,''
  \emph{Neural Information Processing Systems (NeurIPS)}, vol.~34, 2021.

\bibitem{muthirayan2021online}
D.~Muthirayan, J.~Yuan, D.~Kalathil, and P.~P. Khargonekar, ``Online learning
  for predictive control with provable regret guarantees,'' \emph{arXiv
  preprint arXiv:2111.15041}, 2021.

\bibitem{han2022learning}
Y.~Han, R.~Solozabal, J.~Dong, X.~Zhou, M.~Takac, and B.~Gu, ``Learning to
  control under time-varying environment,'' \emph{arXiv preprint
  arXiv:2206.02507}, 2022.

\bibitem{baby2022optimal}
D.~Baby and Y.-X. Wang, ``Optimal dynamic regret in lqr control,'' \emph{arXiv
  preprint arXiv:2206.09257}, 2022.

\bibitem{hazan2008adaptive}
E.~Hazan, A.~Rakhlin, and P.~L. Bartlett, ``Adaptive online gradient descent,''
  \emph{Advances in Neural Information Processing Systems}, pp. 65--72, 2008.

\bibitem{zhao2022non}
P.~Zhao, Y.-X. Wang, and Z.-H. Zhou, ``Non-stationary online learning with
  memory and non-stochastic control,'' pp. 2101--2133, 2022.

\bibitem{zhao2020simple}
P.~Zhao, L.~Zhang, Y.~Jiang, and Z.-H. Zhou, ``A simple approach for
  non-stationary linear bandits,'' pp. 746--755, 2020.

\bibitem{muthirayan2022adaptive}
D.~Muthirayan, R.~Du, Y.~Shen, and P.~P. Khargonekar, ``Adaptive control of
  unknown time varying dynamical systems with regret guarantees,'' \emph{arXiv
  preprint arXiv:2210.11684}, 2022.

\bibitem{abbasi2011regret}
Y.~Abbasi-Yadkori and C.~Szepesv{\'a}ri, ``Regret bounds for the adaptive
  control of linear quadratic systems,'' \emph{Proceedings of the 24th Annual
  Conference on Learning Theory}, pp. 1--26, 2011.

\bibitem{oymak2019non}
S.~Oymak and N.~Ozay, ``Non-asymptotic identification of lti systems from a
  single trajectory,'' \emph{2019 American control conference (ACC)}, pp.
  5655--5661, 2019.

\bibitem{abbasi2011improved}
Y.~Abbasi-Yadkori, D.~P{\'a}l, and C.~Szepesv{\'a}ri, ``Improved algorithms for
  linear stochastic bandits,'' \emph{Advances in neural information processing
  systems}, vol.~24, 2011.

\end{thebibliography}
